\documentclass[reprint,aps,showpacs,amssymb,amsfonts,superscriptaddress,twocolumn,pra,nofootinbib,longbibliography]{revtex4-1}
\usepackage{bbm,mathrsfs}
\usepackage{graphicx}
\usepackage[utf8]{inputenc}
\usepackage{amsfonts,amsmath}
\usepackage{amsthm}
\usepackage{times}
\usepackage{longtable}
\usepackage{changes}

\def\textbf#1{{\bf #1}}
\def\be{\begin{equation}}
\def\ee{\end{equation}}
\def\ben{\begin{eqnarray}}
\def\een{\end{eqnarray}}
\def\eea{\end{array}}
\def\bea{\begin{array}}
\newcommand{\Tr}[0]{\mathrm{Tr}}
\newcommand{\ot}[0]{\otimes}
\newcommand{\bei}{\begin{itemize}}
\newcommand{\eei}{\end{itemize}}
\newcommand{\ket}[1]{|#1\rangle}
\newcommand{\bra}[1]{\langle#1|}

\definecolor{myurlcolor}{rgb}{0,0,0.7}
\definecolor{myrefcolor}{rgb}{0.8,0,0}
\usepackage{hyperref}
\hypersetup{colorlinks, linkcolor=myrefcolor,
citecolor=myurlcolor, urlcolor=myurlcolor}

\newtheorem{thm}{Theorem}

\newtheorem{lemma}{Lemma}
\theoremstyle{definition}

\begin{document}

\title{Energy as a detector of nonlocality of many-body spin systems}

\author{J. Tura}

\affiliation{ICFO - Institut de Ciencies Fotoniques, The Barcelona Institute of Science and Technology, 08860 Castelldefels (Barcelona), Spain}
\affiliation{Max-Planck-Institut f\"ur Quantenoptik, Hans-Kopfermann-Stra{\ss}e 1, 85748 Garching, Germany}

\author{G. De las Cuevas}

\affiliation{Max-Planck-Institut f\"ur Quantenoptik, Hans-Kopfermann-Stra{\ss}e 1, 85748 Garching, Germany}
\affiliation{Institut f\"ur theoretische Physik, Universit\"at Innsbruck, Technikerstr.\ 21a, 6020 Innsbruck, Austria}

\author{R. Augusiak}
\affiliation{Center for Theoretical Physics, Polish Academy of Sciences, Aleja Lotnik\'ow 32/46, Warsaw, Poland}
\author{M. Lewenstein}
\affiliation{ICFO - Institut de Ciencies Fotoniques, The Barcelona Institute of Science and Technology, 08860 Castelldefels (Barcelona), Spain}
\affiliation{ICREA - Instituci\'o Catalana de Recerca i Estudis Avan\c{c}ats, 08011 Barcelona, Spain}
\author{A. Ac\'in}
\affiliation{ICFO - Institut de Ciencies Fotoniques, The Barcelona Institute of Science and Technology, 08860 Castelldefels (Barcelona), Spain}

\affiliation{ICREA - Instituci\'o Catalana de Recerca i Estudis Avan\c{c}ats, 08011 Barcelona, Spain}

\author{J. I. Cirac}
\affiliation{Max-Planck-Institut f\"ur Quantenoptik, Hans-Kopfermann-Stra{\ss}e 1, 85748 Garching, Germany}

\begin{abstract}

We present a method to show that low-energy states of quantum many-body interacting systems in one spatial dimension are nonlocal. We assign a Bell inequality to the Hamiltonian of the system in a natural way and we
efficiently find its classical bound using dynamic programming. The Bell inequality is such that its quantum value for a given state, and for
appropriate observables, corresponds to the energy of the state. Thus, the presence of nonlocal correlations can be certified for states of low enough energy. The method can also be used to optimize certain Bell inequalities: in the translationally invariant (TI) case, we provide an exponentially faster computation of the classical bound and analytically closed expressions of the quantum value for appropriate observables and Hamiltonians. The power and generality of our method is illustrated through four representative examples: a tight TI inequality for 8 parties, a quasi TI uniparametric inequality for any even number of parties, ground states of spin-glass systems, and a non-integrable interacting XXZ-like Hamiltonian. Our work opens the possibility for the use of low-energy states of commonly studied Hamiltonians as multipartite resources for quantum information protocols that require nonlocality.
\end{abstract}

 \pacs{}

\maketitle

\section{Introduction}
Nonlocality is a fundamental property of Nature in which the statistics obtained by performing some local measurements on some composite quantum systems cannot be reproduced by any local hidden variable model \cite{bell64}. These so-called nonlocal correlations cannot be mimicked by any local deterministic strategy, even if assisted by shared randomness \cite{Fine82}. Nonlocality is detected by the violation of a Bell inequality \cite{ReviewNonlocality} and it has recently been demonstrated in three loophole-free Bell experiments \cite{BellExperimentDelft, BellExperimentVienna, BellExperimentIllinois}. Detection of nonlocality is a sufficient condition to demonstrate, in a device-independent (DI) way, that the state producing such correlations is entangled. From an operational point of view, nonlocality is a resource that enables the implementation of DI quantum information protocols, such as DI quantum key distribution \cite{AcinQuantumCrypto, DIQKDAcin}, DI randomness expansion \cite{DIRE} or amplification \cite{DIRA, DIRAAcin}, or DI self-testing \cite{Self-Testing, STXupitx}.

The study of quantum many-body systems has benefited during the last decades from insights of the field of quantum information; in particular concerning the understanding of their correlations \cite{RMPEntanglement, RMPUltracold}. This has, however, mostly focused on the study and experimental detection of entanglement \cite{EntanglementCollisions, EntanglementNegativeWigner3kAtoms, SpinSqueezedBEC}, while the role of nonlocal correlations, which are stronger, remains rather unexplored. There are at least three reasons for that: First, the known Bell inequalities for multipartite systems involve correlations among many particles \cite{NonlocalityNBosons, Svetlichny, ZukowskiBrukner, WernerWolf, BIGraphStates}, thus rendering their measurement a formidably challenging task. Second, the mathematical characterization of nonlocal correlations is an NP-hard problem \cite{Babai}. Third, the size of the description of multipartite quantum states grows, in general, exponentially with the system size, posing a strong barrier to the analysis of the quantum correlations in large systems. However, recent advances \cite{TASVLA, AnnPhys} have shown that, by measuring only one- and two-body correlation functions, nonlocality can be revealed in some multipartite quantum systems, opening the way to its detection in many-body systems \cite{NonlocalityBEC} (see also \cite{PezzePRA}).

In this work we show that the ground states of some quantum spin Hamiltonians in one spatial dimension are nonlocal. We assign a Bell inequality to the given Hamiltonian in a natural way and we calculate its classical bound using dynamic programming. The Bell inequality is constructed in such a way that, for appropriate quantum observables, the Bell operator coincides with the Hamiltonian. The idea is that if the ground state energy is beyond the classical bound, this signals the presence of nonlocal correlations in the ground state. The ground state energy is computed by exact diagonalization using the Jordan-Wigner (JW) transformation, which maps the system of spins to a quadratic system of fermions. The method just presented can also be seen from the opposite point of view, namely, as a way to optimize certain classes of Bell inequalities for many-body systems under some quantum observables.
We also study the translationally invariant (TI) setting, in which we provide an exponentially faster algorithm to find the classical bound and we obtain analytical results for the quantum value\footnote{In this work, we refer to the \textit{quantum value} as the expectation value of the ground state of the Bell operator under appropriate quantum observables (such that it matches the Hamiltonian). This should not be confused with the \textit{quantum bound} of a Bell inequality, which is the infimum over all quantum states and observables}.
Then we illustrate our framework by applying it to four representative examples:
First, a tight TI Bell inequality for $8$ parties. Second, a quasi TI Bell inequality for any even number of parties. Third, we show that the ground state of an XY spin glass has nonlocal correlations in some parameter region. Finally, a non-integrable interacting XXZ-like Hamiltonian to which we assign a variation of Gisin's Elegant Bell inequality \cite{GisinElegant} and we find its quantum value numerically using matrix product states and the density matrix renormalization group \cite{iTensor}. This shows that our method is not limited to Hamiltonians that can be mapped to a system of free fermions via the JW transformation, but it can be applied to any spin Hamiltonian with short-range interactions in one spatial dimension.


The paper is organized as follows: In Section \ref{sec:method} we present our method. In Section \ref{sec:otherway} we show how to optimize certain classes of many-body Bell inequalities. In Section \ref{sec:ti} we particularize our results to the TI case. In Section \ref{sec:examples} we illustrate our methods with four examples. Finally, in Section \ref{sec:conclusion} we conclude and explore future lines of research.

\section{The Method}
\label{sec:method}
In this section we present a method to analyze when the ground state of some spin Hamiltonians is nonlocal; namely, when the quantum value is beyond the classical bound. We first present the setting (Section \ref{subsec:thesetting}). We then construct a Bell inequality from the given Hamiltonian (Section \ref{subsec:ConstructingBI}) and we compute its classical bound using dynamic programming (Section \ref{subsec:classical}). Then we find the quantum value of the inequality (Section \ref{subsec:quantumvalue}). To this end, we first review the exact diagonalization of quadratic fermionic Hamiltonians (Section \ref{sec:quantum}) and then we relate the spin Hamiltonian and the fermionic Hamiltonian via the Jordan-Wigner (JW) transformation (Section \ref{subsec:fermion2spins}).

\subsection{The setting}
\label{subsec:thesetting}

We consider quantum spin-$1/2$ Hamiltonians of $n$ particles in one spatial dimension (henceforth, one-dimensional) with periodic boundary conditions, with short-range interactions, up to $R$ neighbors:
\begin{equation}
 {\cal H} = \sum_{i=0}^{n-1} \left(t^{(i)} \sigma_z^{(i)} + \sum_{r=1}^R \sum_{\alpha, \beta=0}^1 t^{(i,r)}_{\alpha, \beta}\mathrm{Str}^{(i,r)}_{\alpha, \beta}\right),
 \label{eq:SpinHamiltonian}
\end{equation}
where $t^{(i)}$ and $t^{(i,r)}_{\alpha, \beta}$ are real parameters,
\begin{equation}
\mathrm{Str}^{(i,r)}_{\alpha, \beta}:= \sigma_{x+\alpha}^{(i)}\sigma_z^{(i+1)}\cdots \sigma_z^{(i+r-1)} \sigma_{x+\beta}^{(i+r)}
\label{eq:quadraticoperatorslist}
\end{equation}
are the so-called String operators, $\sigma_x^{(i)}$, $\sigma_y^{(i)}$ and $\sigma_z^{(i)}$ are the Pauli Matrices acting on the $i$-th site and the indices of the sites are taken \textit{modulo} $n$. We denote $x+1:=y$ for short.

On one hand, as we explain it in detail in Section \ref{subsec:quantumvalue}, this choice of Hamiltonians is convenient from a purely mathematical perspective, as these/they can be exactly diagonalized via the JW transformation. On the other hand, Hamiltonians of the form (\ref{eq:SpinHamiltonian}) are general enough to include many cases of interest.
For instance, the case $R=1$ corresponds to a one-dimensional spin-$1/2$ Hamiltonian with nearest-neighbors interactions, under a transverse magnetic field:
\begin{equation}
 {\cal H} = \sum_{i=0}^{n-1}\left(t^{(i)}\sigma_z^{(i)}+ \sum_{\alpha, \beta = 0}^1t^{(i,1)}_{\alpha, \beta}\sigma_{x+\alpha}^{(i)}\sigma_{x+\beta}^{(i+1)}\right).
\end{equation}
Nevertheless, our method can also be applied to Hamiltonians with local interactions that do not rely on the String operator structure, at the price of having to compute their ground state energy numerically, as we illustrate it in Example \ref{ex:GisinElegant}.

\subsection{Construction of a Bell inequality}
\label{subsec:ConstructingBI}
In this section we study the classes of Bell inequalities that are relevant for our work. We construct Bell inequalities such that, for some quantum observables, the corresponding Bell operator ${\cal B}$ satisfies
\begin{equation}
 {\cal B} = \beta_C {\mathbbm{1}} + {\cal H},
 \label{eq:BellOpGeneral}
\end{equation}
where  $\beta_C \in \mathbbm{R}$ is the so-called classical bound and ${\cal H}$ is defined as in Eq. (\ref{eq:SpinHamiltonian}).
We use the following convention in writing down Bell inequalities. We want to obtain Bell inequalities of the form $I + \beta_C \geq 0$. The part of the Bell operator that corresponds to $I$ will be the Hamiltonian $\cal H$, and thus, states with low enough energy, lower than $-\beta_C$, will give a violation of the Bell inequality. The classical bound $\beta_C$ is defined as
\begin{equation}
\beta_C = - \min_{\mathrm{LHVM}} I,
\end{equation}
where the minimum is taken over all Local Hidden Variable Models (LHVM) (cf. Eq. \cite{ReviewNonlocality}).
Observe that the quantum state that minimizes the expectation value of ${\cal B}$ is the ground state of ${\cal H}$.

This motivates the study of Bell inequalities in the following scenario: We have $n$ parties with $m$ dichotomic observables with outcomes $\pm 1$ at their disposal. We denote the choices of measurements by $\mathbf{k}=(k_0,\ldots,k_{n-1})$, with $0\leq k_i < m$, and the outcomes they produce by $\mathbf{a}=(a_0,\ldots, a_{n-1})$ with $a_i = \pm 1$. We denote by $P(\mathbf{a}|\mathbf{k})$ the vector of conditional probabilities collected when they perform the Bell experiment. Due to the no-signalling principle, the marginals observed by any subset of parties do not depend on the choices of measurements performed by the rest; thus $P(\{a_i\}_{i\in S}|\{k_i\}_{i\in S})$ is well defined on any subset $S$. In the case of dichotomic measurements, one normally 
works with the \emph{correlators} $M_{\mathbf{k}}^{(i,r)}$,
\begin{equation}
M_{\mathbf{k}}^{(i,r)}:=\sum_{\mathbf{a}} \left(\prod_{j=0}^r a_{i+j}\right)P(\mathbf{a}|\mathbf{k}),
\label{eq:def:correlators}
\end{equation}
where abusing notation we are now denoting $\mathbf{k}=(k_i,\ldots, k_{i+r})$ and $\mathbf{a}=(a_i,\ldots, a_{i+r})$.

If $R>1$, we will consider $m+1$ measurements, due to the $\sigma_z$ inbetween (cf. Eq. (\ref{eq:quadraticoperatorslist})).
The Bell inequalities that are naturally tailored to a Bell operator of the form of Eq. (\ref{eq:BellOpGeneral}) can be written as $I + \beta_C \geq 0$, where
\begin{equation}
 I:=\sum_{i=0}^{n-1} \left(\gamma^{(i)}M_m^{(i,0)} + \sum_{r=1}^R\sum_{k,l=0}^{m-1} \gamma_{k,l}^{(i,r)} M_{(k,m,\ldots,m,l)}^{(i,r)}\right)
 \label{eq:BellIneqGeneral}
\end{equation}
and $\gamma^{(i)}$, $\gamma_{k,l}^{(i,r)}$ are real parameters that depend on the $t^{(i)}$ and $t^{(i,r)}_{\alpha, \beta}$ of Eq. (\ref{eq:SpinHamiltonian}).
Despite the fact that $I$ contains up to $(R+1)$-body terms, its coefficients $\gamma^{(i)}$ and $\gamma_{k,l}^{(i,r)}$ show that it is essentially a $2$-body Bell inequality, since the measurement choice of the parties in the middle of the string is fixed to be $m$, in the sense that the number of coefficients $\gamma^{(i)}$ and $\gamma_{k,l}^{(i,r)}$ is the same as in a $2$-body Bell inequality.
Note that the number of measurements $m$ performed on the $x$-$y$ plane will not affect the form of Eq. (\ref{eq:SpinHamiltonian}). The only measurement that is not performed in this plane is in the $z$ direction; therefore, we treat it as a special case and we say we have $m+1$ measurements.

\subsubsection{The classical optimization}
\label{subsec:classical}
We now describe how to efficiently compute the classical bound $\beta_C$ of the Bell inequalities introduced in Eq. (\ref{eq:BellIneqGeneral}). It is well known that for a generic Bell inequality for $n$ parties, $m$ measurements and $d$ outcomes the classical bound cannot be found efficiently, as it requires solving a linear program with $d^{mn}$ constraints \cite{ReviewNonlocality}. The particular form of the Bell inequalities we are considering, however, allows us to find an algorithm that, in the many-body regime (i.e., for fixed $d$, $m$ and $R$) has $O(n)$ complexity.
For simplicity, we consider a Bell inequality $I$ of a slightly more general form than those of Eq. (\ref{eq:BellIneqGeneral}) and with $d=2$ outcomes,
\begin{equation}
I:=\sum_{i=0}^{n-1}\sum_{r=0}^R\sum_{\mathbf{k}=0}^{m^{r+1}-1} \gamma_{\mathbf{k}}^{(i,r)}M_{\mathbf{k}}^{(i,r)},
\label{eq:BellIneq}
\end{equation}
where $\mathbf{k}=(k_i,\ldots,k_{i+r})$ and $0\leq k_j < m$. After we have presented our method, it will become clear that there is no loss of generality in considering dichotomic measurements, as the result can be straightforwardly generalized to an arbitrary $d$.

To find $\beta_C$, we need to optimize $I$ over all local hidden variable models. By Fine's Theorem \cite{Fine82}, it suffices to optimize $I$ over all deterministic local strategies, in which the correlators $M_{\mathbf{k}}^{(i,r)}$ factorize as
\begin{equation}
M_{\mathbf{k}}^{(i,r)} = \prod_{j=0}^{r+1} M_{k_j}^{(i+j)},
\label{eq:LHVDet}
\end{equation}
where $M_{k_i}^{(i)}$ can be $\pm 1$. Thus,
\begin{equation}
 \beta_C=-\min_{M_{k_i}^{(i)}=\pm 1} I,
 \label{eq:minOBC}
\end{equation}
where the minimum is taken over all possible assignments of $M_{k_i}^{(i)}$ to $\pm 1$ for all $i$ and $k$.

Let us first solve the case with Open Boundary Conditions (OBC); \textit{i.e.}, the case where $\gamma_{\mathbf{k}}^{(i,r)}=0$ when $i+r\geq n$. We shall follow a dynamic programming procedure \cite{ClassicalSpins} that splits the minimization (\ref{eq:minOBC}) into nested parts.

To this end, we represent a local deterministic strategy as a matrix $\mathbf{M}$ whose rows index the measurement choices and whose columns index the party, and the entry $(k,i)$ is the value assigned in the deterministic strategy to the $k$-th observable of the $i$-th party. Thus, $\mathbf{M}$ is a $(m \times n)$ matrix whose entries take integer values $+1$ or $-1$. Let $\mathbf{M}^{(i,R)}$ denote the submatrix of $\mathbf{M}$ consisting of columns $i$ to $i+R-1$.
The goal of the dynamic programming is to find an optimal $\mathbf{M}$, which will give $\beta_C$. This will be obtained recursively.

Let $h_i$ be the function defined for $i>0$ as
\begin{equation}
h_i(\mathbf{M}^{(i-1,R+1)}):=\sum_{r=0}^R\sum_{\mathbf{k}=0}^{m^{r+1}-1} \gamma_{\mathbf{k}}^{(i-1,r)}M_{\mathbf{k}}^{(i-1,r)}.
\label{eq:hsubi}
\end{equation}
Note that because of Eq. (\ref{eq:LHVDet}), $h_i(\mathbf{M}^{(i-1,R+1)})$ is a real number.
Now we define a recursive function $E_i$ which contains the optimization up to the $(i-1)$-th site. Explicitly, $E_0(\mathbf{M}^{(0,R)}):=0$ and
\begin{equation}
 E_i(\mathbf{M}^{(i,R)}):=\min_{M_k^{(i-1)}} \left\{E_{i-1}(\mathbf{M}^{(i-1,R)})+h_i(\mathbf{M}^{(i-1,R+1)})\right\},
 \label{eq:Ei}
\end{equation}
for $i>0$.
Note that $E_i$ optimizes the local deterministic strategy on the $(i-1)$-th party, which amounts to choosing the optimal values of the $(i-1)$-th column of $\mathbf{M}$. This naturally depends on the next $R$ columns, which we need to consider as variables and calculate $E_i$ for all their possible values.
Therefore, to efficiently evaluate $E_i(\mathbf{M}^{(i,R)})$, we only need to access the stored values of $E_{i-1}$ on $\mathbf{M}^{(i-1,R)}$ and not $E_{i-2}$ and so on. The computation of $E_i(\mathbf{M}^{(i,R)})$ thus requires evaluation of the $2^m$ different deterministic local strategies corresponding to the values $M_k^{(i-1)}$ (see Figure \ref{fig:DynProg}).

\begin{center}
\begin{figure}
 \includegraphics[width=0.9\columnwidth]{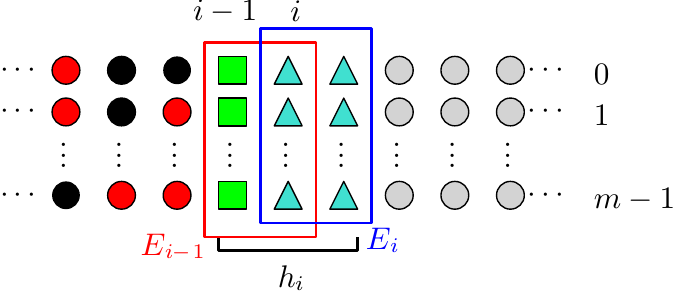}
\caption{The representation of the matrix $\mathbf{M}$ at the step where we compute $E_i$ for an inequality with dichotomic measurements with $R=2$. 
The color (red or black) corresponds to the outcome assigned to that observable. At the $i$-th step, an optimal local deterministic strategy for parties $0$ to $i-2$ is already set, and we are minimizing the local deterministic strategy at the $(i-1)$-th site (green squares). In order to perform the minimization, we need the values of the $i$-th and $(i+1)$-th parties to be fixed (blue triangles); i.e., we need to give a value to $\mathbf{M}^{(i,2)}$. We check the $2^{m}$ possible values for the observables at site $i-1$ and we find an optimal assignment for them.}
\label{fig:DynProg}
\end{figure}
\end{center}

The classical bound (\ref{eq:minOBC}) is obtained as the end product of this minimization procedure, namely 
\begin{equation}
\beta_C=-E_n(\mathbf{M}^{(n,R)}). 
\end{equation}

Note that $E_n$ is actually independent of $\mathbf{M}^{(n,R)}$ because we are in the OBC case, so all the $\gamma_{\mathbf{k}}$'s that would extend beyond the $n$-th party are zero.

This procedure can be easily generalized if the Bell inequality has $d>2$ outcomes. In this case, one has to take into account that the notion of correlators introduced in Eq. (\ref{eq:def:correlators}) is no longer well defined (as the one-to-one correspondence between probabilities $P(\mathbf{a}|\mathbf{k})$ and correlators $M_{\mathbf{k}}^{(i,r)}$ no longer holds). Thus, one needs to express the Bell inequality in terms of probabilities:
\begin{equation}
I = \sum_{i=0}^{n-1}\sum_{r=0}^R \sum_{\mathbf{k}=0}^{m^{r+1}-1} \sum_{\mathbf{a}=0}^{d^{r+1}-1} \gamma_{\mathbf{k}, \mathbf{a}}^{(i,r)}P(\mathbf{a}|\mathbf{k}),
\end{equation}
where $\mathbf{a}=(a_i,\ldots,a_{i+r})$ and $\mathbf{k}$ has the same structure.
%
Note that $P(\mathbf{a}|\mathbf{k})$ implicitly depends on $i$ and $r$ (cf. Eq. (\ref{eq:def:correlators})), thus being a marginal $(r+1)$-body probability distribution. Again, like in the $d=2$ case, by virtue of Fine's theorem \cite{Fine82}, it suffices to consider those probability distributions in which $P(\mathbf{a}|\mathbf{k})$ factorizes; \textit{i.e.,}
\begin{equation}
P(\mathbf{a}|\mathbf{k}) = \prod_{j=0}^{r} P(a_{i+j}|k_{i+j}).
\end{equation}
Since each $P(a_i|k_i)$ can now take $d$ different deterministic values, the minimization in Eq. (\ref{eq:Ei}) is carried over variables that can take up to $d$ different values. For instance, in Fig. \ref{fig:DynProg}, instead of only red and black, one would have $d$ different possible colors.

In summary, the overall minimization is performed in $O(n d^{m(R+1)})$ time. Since $d$, $m$ and $R$ are fixed, the overall scaling is $O(n)$, although in practice it is advisable to bear in mind the prefactor. This algorithm gives not only the classical bound $\beta_C$ (for which it requires $O(d^{mR})=O(1)$ memory) but it also constructs a deterministic local strategy achieving it (requiring $O(mn)=O(n)$ memory).

Let us now consider the Periodic Boundary Conditions (PBC) case. This can be reduced to the OBC case by splitting the Bell inequality with PBC at an arbitrary position $i$ while fixing a value of $\mathbf{M}^{(i,R)}$.
The correlators that contain parties in the set $\{i, \ldots i+R\}$ can be effectively moved to the left/right of the cut by updating the coefficients of $I$: If all the parties on the correlator belong to the set $\{i,\ldots,i+R\}$, then this correlator has a definite value which becomes a constant in the optimization; if just one party lies outside, then the one-body weights at sites $i-1$ or $i+R+1$ can be updated accordingly, and so on. In Appendix \ref{app:PBC} we explain this procedure in detail.

Since the amount of $\mathbf{M}^{(i,R)}$'s for which the actual minimum of $I$ is achieved is finite, 
the PBC case is solved by considering 
$d^{m R}$ OBC cases. This does not change the overall complexity, but it increases the prefactor.
The classical bound of $I$ is in the PBC case found
in $O(nd^{m(2R+1)}) = O(n)$ time and $O(1)$ memory for $\beta_C$ and $O(n)$ memory for the deterministic local strategy.

Note that it is crucial for our method that $R$ is constant. If $R$ were comparable to $n$, the dynamic programming procedure would no longer be efficient. In fact, optimizing one-dimensional Bell inequalities with full-range correlators is equivalent to optimizing general Bell inequalities, where results in computer science indicate that this is an extremely hard problem \cite{ReviewNonlocality, Babai}. Note that, even in the bipartite case, deciding whether a probability distribution for $m$ dichotomic measurements is local is NP-complete \cite{Avis2004}.

\subsection{The quantum value}
\label{subsec:quantumvalue}
In this section, we show how to find the ground state energy of the Hamiltonian ${\cal H}$ introduced in Eq. (\ref{eq:SpinHamiltonian}). To do that, we first review the exact diagonalization of quadratic Hamiltonians in fermionic operators (Section \ref{sec:quantum}). Then, we transform the spin operator ${\cal H}$ to a quadratic fermionic operator $\hat{\cal H}$ via the JW transformation \cite{JWTrafo}; see also \cite{NielsenNotesJW, EntanglementFermionicSystems} (Section \ref{subsec:fermion2spins}).

Throughout this paper, we will denote fermionic operators with a hat and spin operators without. Note that the JW transformation is a global operation that breaks the sense of locality in a Bell experiment, but we use it merely as a mathematical tool to find the ground state energy of $\cal H$.

Recall that a quadratic Hamiltonian in fermionic operators has the form
\begin{equation}
 {\hat{\cal H}}=\sum_{0\leq i\neq j < n} A_{ij}\hat{a}_i^\dagger \hat{a}_j  + B_{ij}\hat{a}_i \hat{a}_j + \mathrm{h.c.},
 \label{eq:QuadraticAs}
\end{equation}
where $A_{ij}$ and $B_{ij}$ are complex numbers and $\mathrm{h.c.}$ stands for hermitian conjugate.
The $\hat{a}_i$ ($\hat{a}_i^\dagger$) are annihilation (creation) operators of a fermionic system of $n$ Dirac modes, indexed by $i$, with $0 \leq i < n$.
The $i$-th mode is assigned an annihilation operator $\hat{a}_i$ and a creation operator $\hat{a}_i^\dagger$. The creation operator $\hat{a}_i^\dagger$, acting on the vacuum state $\ket{\Omega}$, populates it with a single excitation: $\ket{1_i}:=\hat{a}_i^\dagger \ket{\Omega}$, whereas the annihilation operator satisfies $\hat{a}_i\ket{\Omega}=0 \ \forall i$. Such operators satisfy the following canonical anticommutation relations (CARs):
\begin{equation}
 \{\hat{a}_i, \hat{a}_j^\dagger\}=\delta_{i,j}{\hat{\mathbbm{1}}},\quad \{\hat{a}_i, \hat{a}_j\}=0 \qquad \forall i,j.
 \label{eq:CARsDirac}
\end{equation}
The CARs (\ref{eq:CARsDirac}) imply that, without loss of generality, the matrices $A$ and $B$ in Eq. (\ref{eq:QuadraticAs}) can be taken to be Hermitian and antisymmetric, respectively.

\subsubsection{Exact diagonalization}

\label{sec:quantum}
In this section we compute the ground state energy of a Hamiltonian of the form (\ref{eq:QuadraticAs}).
The operator $\hat{{\cal H}}$ of Eq. (\ref{eq:QuadraticAs}) can be written in terms of Majorana fermions as
\begin{equation}
 \hat{{\cal H}}=\frac{\mathbbm{i}}{2}\sum_{i,j=0}^{n-1}\sum_{\alpha, \beta = 0}^1 H_{i,\alpha;j, \beta}\hat{c}_{i,\alpha}\hat{c}_{j,\beta},
 \label{eq:QuadraticHamiltonian}
\end{equation}
where the $2n$ Majorana fermions $\hat{c}_{i,\alpha}$ are defined as
\begin{equation}
 \hat{c}_{i,\alpha}:=\mathbbm{i}^\alpha (\hat{a}_i + (-1)^\alpha \hat{a}_i^\dagger ), \quad \alpha \in \{0,1\},\ 0 \leq i < n,
\end{equation}
where $\mathbbm{i}^2+1=0$.
Note that $\hat{c}_{i,\alpha}$ are Hermitian operators, and they satisfy the CARs
\begin{equation}
 \{\hat{c}_{i,\alpha}, \hat{c}_{j,\beta}\} = 2 \delta_{i,j}\delta_{\alpha,\beta} \hat{\mathbbm{1}}.
 \label{eq:CARsMajorana}
\end{equation}
It follows that the matrix $H$ in Eq. (\ref{eq:QuadraticHamiltonian}) can be taken real antisymmetric due to Eq. (\ref{eq:CARsMajorana}) without loss of generality. Since every real antisymmetric matrix $H$ admits a Williamson decomposition $H=OJO^T$ \cite{WilliamsonNormalForm}, where $O$ is an orthogonal transformation and
\begin{equation}
 J=\bigoplus_{k=0}^{n-1} \left(
 \begin{array}{cc}
  0&\varepsilon_k\\
  -\varepsilon_k&0
 \end{array}
\right),
\label{eq:Williamson}
\end{equation}
the operator $\hat{{\cal H}}$ is diagonalized by introducing a new set of Majorana operators $\hat{d}_{k,a}$:
\begin{equation}
 \hat{{\cal H}}=\mathbbm{i}\sum_{k=0}^{n-1}\varepsilon_k\hat{d}_{k,0}\hat{d}_{k,1},
 \label{eq:free}
\end{equation}
where $\{\mathbbm{i}\hat{d}_{k,0}\hat{d}_{k,1}\}$ mutually commute and the new Majorana operators are defined as
\begin{equation}
\hat{d}_{k,a}:= \sum_{i,\alpha}O_{i,\alpha; k,a} \hat{c}_{i,\alpha},
\label{eq:orthogonaltrafo}
\end{equation}
satisfying the same CARs as in Eq. (\ref{eq:CARsMajorana}). Note that every orthogonal transformation $O \in {\cal O}(2n)$\footnote{The set of orthogonal matrices of size $n$ is denoted ${\cal O}(n)$.} takes a set of Majorana fermions $\{\hat{c}_{i,\alpha}\}$ into a new set $\{\hat{d}_{k,a}\}$ obeying the same CARs as Majorana fermions.
In Appendix \ref{APP:NumericalOrtho} we discuss details on how to obtain an $O$ for which Eq. (\ref{eq:Williamson}) holds.

The minimal eigenvalue $E_0$ of $\hat{\cal H}$ is then given by
\begin{equation}
E_0=\sum_{k=0}^{n-1} s_k\varepsilon_k,
\label{eq:GndStateEnergy}
\end{equation}
achieved on a simultaneous eigenstate of the operators $\{\mathbbm{i}\hat{d}_{k,0}\hat{d}_{k,1}\}$, with respective eigenvalue $s_k := -\mathrm{sign}(\varepsilon_k)$.

\subsubsection{From spins to fermions}

\label{subsec:fermion2spins}
In this section we review the JW transformation and we show that the spin Hamiltonians that can be diagonalized with the method described in Section \ref{sec:quantum} are precisely those of the form of Eq. (\ref{eq:SpinHamiltonian}).

The JW transformation establishes an isomorphism between the Fock space of $n$ Majorana modes and the $n$-qubit Hilbert space $(\mathbbm{C}^2)^{\otimes n}$. For Majorana fermions, the JW transformation can be expressed as
\begin{equation}
 \hat{c}_{i,\alpha} \leftrightarrow (-1)^\alpha \left(\prod_{j=0}^{i-1} \sigma_z^{(j)}\right) \sigma_{x+\alpha}^{(i)}, \quad \alpha \in \{0,1\}.
 \label{eq:JW}
\end{equation}
It follows that for every fermionic operator we obtain an operator acting on $({\mathbbm C}^2)^{\otimes n}$ and viceversa.
The operator $\hat{\cal H}$ as in Eq. (\ref{eq:QuadraticHamiltonian}) consists of the terms $\mathbbm{i}\hat{c}_{i,0}\hat{c}_{i,1}$, which correspond to $\sigma_z^{(i)}$, and 
\begin{equation}
 \mathbbm{i} (-1)^\beta \hat{c}_{i,1-\alpha}\hat{c}_{i+r,\beta}, \quad \alpha, \beta \in \{0,1\}, r\geq 1,
 \label{eq:quadraticoperatorslistfermions}
\end{equation}
which correspond to the string operators $\mathrm{Str}_{\alpha, \beta}^{i,r}$ (cf. Eq. (\ref{eq:quadraticoperatorslist})) if $i+r<n$.
If $i+r\geq n$, we need to define a global parity operator
\begin{equation}
 \hat{\mathbbm{P}}:=
 \prod_{i=0}^{n-1}\mathbbm{i}\hat{c}_{i,0}\hat{c}_{i,1} \leftrightarrow \prod_{i=0}^{n-1}\sigma_z^{(i)}=:\mathbbm{P}.
 \label{eq:def:P}
\end{equation}
In this case, the string operator takes the form
\begin{equation}
\begin{array}{lcl}
(-1)^{1+\alpha + \beta}\mathrm{Str}_{1-\beta, 1-\alpha}^{i+r-n,n-r} =  \mathrm{Str}_{\alpha, \beta}^{i,r} \ \mathbbm{P}.
\end{array}
\label{eq:POderNoPThatsthequestion}
\end{equation}

Note that, since $\cal H$ commutes with $\mathbbm{P}$, its ground state has a well defined eigenvalue $p = \pm 1$ of the parity. This fact will become relevant in Section \ref{sec:otherway}. 

\section{A new method to optimize Bell inequalities}
\label{sec:otherway}

The method described so far can also be seen from the opposite perspective: 
One considers Bell inequalities of the form of Eq. (\ref{eq:BellIneqGeneral}) with given coefficients $\gamma_{\mathbf{k}}^{(i,r)}$'s and one calculates its $\beta_C$ with dynamic programming.
On the other hand, we optimize the quantum value over a restricted set of measurements, namely, single qubit $\sigma_z$ measurements, and for the $r$-body correlators, qubit measurements in the $x-y$ plane for parties $i$ and $i+r$ and $\sigma_z$ measurements in the intermediate ones (cf. Eq. (\ref{eq:BellIneqGeneral})). The resulting Bell operator can be mapped to a system of fermions as in Eq. (\ref{eq:QuadraticHamiltonian}), whose diagonalization allows us to find the minimal quantum value.

The latter process is carried out as follows. Let us denote the $k$-th observable of the $i$-th party by ${\mathcal M}_k^{(i)}$, with $k$ ranging from $0$ to $m-1$ (or $m$ if $R>1$). We shall pick qubit observables parametrized as ${\mathcal M}_m^{(i)} = \sigma_z^{(i)}$ and ${\mathcal M}_k^{(i)} = \cos \varphi_{k}^{(i)} \sigma_x^{(i)} + \sin \varphi_{k}^{(i)} \sigma_y^{(i)}$ ($k < m$) and construct the Bell operator ${\cal B}$, which will be of the form (\ref{eq:BellOpGeneral}). 
Note that to build the Bell operator, one simply needs to substitute the correlators in (\ref{eq:BellIneqGeneral}) by the corresponding quantum observables:
\begin{equation}
M_{\mathbf{k}}^{(i,r)} \rightarrow \bigotimes_{j=0}^{r+1} {\mathcal M}_{k_j}^{(i+j)}.
\end{equation}
If there exists a quantum state $\rho$ for which $\Tr({\cal B} {\rho}) < 0$, then $\rho$ is nonlocal. 
If this is the case, we shall denote the quantum violation observed by $QV := \mathrm{Tr}({\cal B} {\rho})$. 
Note that $\rho$ can be taken, without loss of generality, as a projector onto a ground state of ${\cal H}$. 
We will look for the optimal measurement settings ${\varphi_k^{(i)*}}$ such that ${\cal B}$ has the minimal eigenvalue. 

The spin Hamiltonian $\cal H$ is again diagonalized by applying the JW transformation (cf. Eq. (\ref{eq:JW})), and Eq. (\ref{eq:SpinHamiltonian}) is almost transformed into Eq. (\ref{eq:QuadraticHamiltonian}) up to the 
string operators that cross the origin, which carry a parity operator $\hat{\mathbbm{P}}$.
Since the eigenstate with the lowest eigenvalue of $\hat{\cal H}$ has a well defined parity $p=\pm 1$ (because $[\hat{\cal H},\hat{\mathbbm{P}}]=0$), we can change $\mathbbm{P}$ by $p$ in Eq. (\ref{eq:POderNoPThatsthequestion}) so that $\hat{\cal H}$ is now quadratic. One has to make sure, though, that the superselection rule imposed by initially choosing $p$ is obeyed. That is, the ground state of $\hat{\cal H}$ needs to satisfy
\begin{equation}
p=(\det O) \prod_{k=0}^{n-1}s_k.
\label{eq:superselectionrule}
\end{equation}
Eq. (\ref{eq:superselectionrule}) stems from the fact that, under the transformations of Eq. (\ref{eq:orthogonaltrafo}), $\hat{\mathbbm{P}}$ is transformed as (see Appendix \ref{APP:parity} for a proof)
\begin{equation}
 \hat{\mathbbm{P}} =  (\det{O})\prod_{i=0}^{n-1}\mathbbm{i}\hat{d}_{i,0}\hat{d}_{i,1}.
 \label{eq:determinant}
\end{equation}
If Eq. (\ref{eq:superselectionrule}) does not hold, one has to modify Eq. (\ref{eq:GndStateEnergy}) accordingly by picking
\begin{equation}
E_0\rightarrow E_0 + 2 \min_k |\varepsilon_k|.
\end{equation}
The minimal $E_0$ for $p=1$ or $p=-1$ is the minimal eigenvalue of ${\cal B}$.

Finally, note that if $R=1$ and $m=2$, the minimal eigenvalue of ${\cal B}(\varphi_k^{(i)*})$ yields the minimal value of $I$ achievable within quantum theory, denoted $-\beta_Q$,
because the optimal quantum violation of Bell inequalities with $n$ parties and two dichotomic observables is obtained with qubits and traceless observables on a plane \cite{TonerVerstraete}. However, if $R>1$ or $m>2$ this result does not hold in general, as higher-dimensional systems and more general observables can produce more nonlocal correlations.

\section{The translationally invariant case}
\label{sec:ti}
In this section, 
we consider the case in which 
$\cal H$ in Eq. (\ref{eq:SpinHamiltonian}) is translationally invariant (TI).
In the spirit of Section \ref{sec:otherway}, the following procedure can be seen as an optimization of a TI Bell inequality with the same set of observables at each site.
We first present an algorithm to compute the classical bound of a TI Bell inequality with short range correlators which is exponentially faster in the number of parties than the one of Section \ref{subsec:classical} (Section \ref{subsec:FastDP}).
We then find the ground state energy of a TI Hamiltonian of the form of Eq. (\ref{eq:SpinHamiltonian}) analytically (Section \ref{subsec:AnalyticalTI}).

\subsection{Exponentially faster solution of the classical bound}
\label{subsec:FastDP}
In this section, we present a method that is exponentially faster in the number of particles than the one in Section \ref{subsec:classical} to compute the classical bound of TI Bell inequalities of the form of Eq. (\ref{eq:BellIneq});
i.e., where $\gamma_{\mathbf{k}}^{(i,r)}$ are independent of $i$. For the sake of simplicity, we first present this problem in a more abstract setting. Then we adapt it to TI Bell inequalities.

Consider a function $f^{(0)}: S \times S \longrightarrow \mathbbm{R}$ where $S$ is a finite set. We will describe how to compute
\begin{equation}
F:=\min_{x_0,\ldots, x_{w}} \sum_{j=0}^{w-1} f^{(0)}(x_j, x_{j+1})
\end{equation}
in $O(\log_2 w)$ steps. To this end, let us define
\begin{equation}
f^{(t+1)}(x,y) := \min_{z} (f^{(t)}(x,z)+f^{(t)}(z,y))
\label{eq:iterationstepf}
\end{equation}
for $t>0$. Note that the superscript $t$ indicates the iteration step. The idea is to successively rewrite $F$ in terms of $f^{(t+1)}$ instead of $f^{(t)}$ thus eliminating at each step approximately half of the variables in the minimization. For instance, by writing $F$ in terms of $f^{(1)}$ one has already carried out the minimization over all the variables with odd index (except the last one if $w$ is odd). Note that any function $f^{(t)}$ is defined by specifying $|S|^2$ numbers, where $|S|$ is the number of elements of $S$. Computing Eq. (\ref{eq:iterationstepf}) requires $O(|S|^3)$ operations. Note that when $w$ is a power of $2$, then $F$ is given by
\begin{equation}
F=\min_{x_0,x_w} f^{(\log_2 w)}(x_0,x_w).
\end{equation}
In the general case, however, we cannot assume $w$ to be a power of $2$. Nevertheless, every positive integer $w$ can be uniquely expressed as a sum of different powers of $2$. The idea is to apply the procedure described above to each of these powers of $2$ and then optimize over the remaining $O(\log_2 w)$ variables. To this end, recall that the binary expression for $w$ is
\begin{equation}
w=\sum_{i=0}^{\lfloor\log_2 w \rfloor} a_i 2^i= \sum_{j=0}^{|w|-1}2^{b_j},
\end{equation}
where $\left\lfloor \cdot \right\rfloor$ is the floor function, $a_i \in \{0,1\}$ correspond to the digits of $w$ in binary, $|w|:=\sum_{i} a_i$ is the Hamming weight of $w$, and the $b_j$'s enumerate the indices $i$ for which $a_i = 1$, sorted in decreasing order.
For instance, if $w=11$, $(a_3,a_2,a_1,a_0)=(1,0,1,1)$ and $(b_0,b_1,b_2)=(3,1,0)$, and if $w=15$, then $(a_3,a_2,a_1,a_0) = (1,1,1,1)$ and $(b_0,b_1,b_2,b_3)=(3,2,1,0)$.

Note that $F$ can now be rewritten as
\begin{equation}
F=\min_{y_{0}\ldots y_{|w|}} \sum_{j=0}^{|w|-1} f^{(b_j)} (y_j, y_{j+1}),
\label{eq:F:LastStep}
\end{equation}
which is a minimization over $1+|w|\leq 1+\lceil\log_2(w)\rceil$ variables $y_j$, where $\lceil \cdot \rceil$ is the ceiling function. Note that the $f^{(b_j)}$ need no longer be the same function for different $j$, so the expression for $F$ given in Eq. (\ref{eq:F:LastStep}) is not TI in general.
Note that at the $\lfloor \log_2 w\rfloor$-th step of the optimization we have the $x_0$ variable on the left (see Fig. \ref{fig:CampesinoRuso}) and several remaining variables $x_j$ for $j \geq \lfloor \log_2 w\rfloor$. Since these remaining variables highly depend on the binary expression of $w$, we denote them by $y_j$ (cf. Eq. (\ref{eq:F:LastStep})).
To minimize over $y_j$, we proceed in a similar fashion, now defining $g^{(0)}:=f^{(b_0)}$ and
\begin{equation}
g^{(t+1)}(x,y):=\min_z \left( g^{(t)}(x,z)+f^{(b_{t+1})}(z,y) \right)
\end{equation}
for $t>0$. Observe that now the optimization is linear, similar to the dynamic programming presented in Section \ref{subsec:classical}.
It follows that $F$ can now be computed as
\begin{equation}
F=\min_{x_0,x_{w}}g^{(|w|-1)}(x_0,x_w).
\end{equation}
In Figure \ref{fig:CampesinoRuso} we describe this procedure with an example.
\begin{center}
\begin{figure}
 \includegraphics[width=0.9\columnwidth]{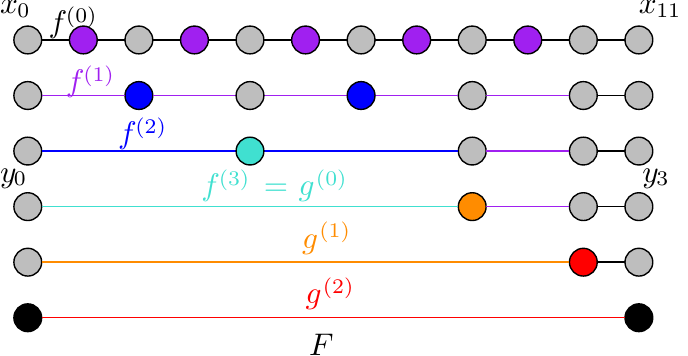}
\caption{Here we find $F$ for $w=11$. We represent the variables $x_i$ with circles and the functions $f$ or $g$ with lines. Each line with the same color corresponds to the same function. The variables that are not in gray are eliminated at the next step. At the $0$-th iteration, there are $10$ functions $f^{(0)}$ depending on $11$ variables. We compute $f^{(1)}$ and we substitute it as many times as possible, so that at the $1$-st iteration, we have eliminated the variables in purple. Then we eliminate the variables in blue by computing $f^{(2)}$ and the variables in turqouise by computing $f^{(3)}$. We then compute the functions $g$ by joining them with the remaining $f$'s, thus eliminating the orange and the red variables. Finally we minimize on the ends, where we can impose conditions on the boundary if needed.}
\label{fig:CampesinoRuso}
\end{figure}
\end{center}

To adapt this algorithm to the minimization of $I$, we start by noting that $I$ can be written as
\begin{equation}
I=\sum_{i=0}^{n-1} h(\mathbf{M}^{(i,R+1)}),
\end{equation}
where $h$ is defined as
\begin{equation}
h(\mathbf{M}^{(i,R+1)}):=\sum_{r=0}^R\sum_{\mathbf{k}=0}^{m^{r+1}-1}\gamma_\mathbf{k}^{(r)}M_{\mathbf{k}}^{(i,r)},
\end{equation}
and the indices of the parties are taken \textit{modulo} $n$ in the $M_{\mathbf{k}}^{(i,r)}$ defined in Eq. (\ref{eq:LHVDet}). Observe that every $i$-th and $(i+R+1)$-th parties share $R$ parties via $h$. In particular, by picking $i = j(R+1)-1$ we denote their local deterministic strategy as
\begin{equation}
x_j := \mathbf{M}^{(j(R+1), R)}.
\end{equation}
We now rewrite the optimization of $I$ over $\mathbf{M}$ in terms of $x_j$.
To this end, let us define
\begin{equation}
f^{(0)}(x_j, x_{j+1}):=\min_{M_k^{(jR+j+R)}} \sum_{i=0}^{R} h(\mathbf{M}^{(j(R+1)+i,R+1)}).
\end{equation}
Since we cannot assume that $n$ is a multiple of $R+1$, we take $w:=\lfloor n/(R+1)\rfloor$. Then,
\begin{equation}
\min_{\mathbf{M}} I = \widetilde{\min_{x_0, x_w}} g^{(|w|-1)}(x_0, x_w)+ T(x_w,x_0).
\end{equation}
where the tail $T(x_w, x_0)$
is defined as
\begin{equation}
T(x_w,x_0):=\sum_{i=(R+1)w}^{n-1} h(\mathbf{M}^{(i,R)}),
\end{equation}
with the indices of the parties taken \textit{modulo} $n$ and $\widetilde{\min}$ is the minimum taken on those $x_0, x_w$ that are compatible with PBC (for instance, if $w=n(R+1)$, then $T=0$ and $\widetilde{\min}$ is taken over $x_0=x_w$). In Fig. \ref{fig:DynProgTI} we illustrate the procedure we described with an example.

\begin{center}
\begin{figure}
 \includegraphics[width=0.9\columnwidth]{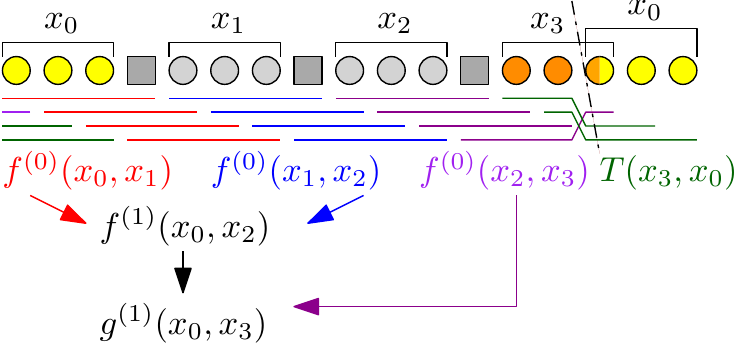}
\caption{An example with $n=14$ and $R=3$ (this corresponds to $w = 3$). Each circle or square corresponds to a party, starting at $i=0$ on the left. The dash-dotted line represents crossing the origin and each line below the parties represents a function $h$ (note that $h$ has a range of $R+1$ parties). The lines corresponding to the $h$'s are arranged in groups of $R+1$ (except for the $T$ corresponding to the tail), which we represent with the same color. Each full group contains a single gray square. By minimizing the local deterministic strategy at the square, we can define $f^{(0)}$ which depends on the local deterministic strategy chosen at the $R$ neighbours of each side. These groups of $R$ parties correspond to the $x_j$ in Fig. \ref{fig:CampesinoRuso}. They encode the possible $d^{mR}$ local deterministic strategies for each $x_j$.
Now, in $O(\log n)$ steps we find $g^{|w|-1}(x_0,x_w)$. Finally, we find the classical bound by minimizing the sum of $g^{|w|-1}$ and $T$ restricted to the $x_0$ and $x_{w}$ that have a compatible overlap.}
\label{fig:DynProgTI}
\end{figure}
\end{center}

\subsection{Analytical solution of the quantum value}
\label{subsec:AnalyticalTI}
Here we consider Eq. (\ref{eq:SpinHamiltonian}) in the TI case, which corresponds to $t^{(i)}$ and $t^{(i,r)}_{\alpha, \beta}$ being independent of $i$. In terms of Bell inequalities, this corresponds to the optimization of TI Bell inequalities with the same set observables being performed at each site.
We give analytically closed expressions in this case. As we prove in Appendix \ref{APP:QuantumTI}, the Williamson eigenvalues for a TI Bell operator of the form of Eq. (\ref{eq:BellOpGeneral}) are given by
\begin{equation}
 \varepsilon_{k,\pm}:=a_k + c_k \pm \sqrt{(a_k-c_k)^2+4(b_k^2+x_k^2)},
 \label{eq:WilliamsonTI}
\end{equation}
with $k$ ranging from $1$ to $\lfloor(n-(p-1)/2)/2 \rfloor$, 
where
\begin{eqnarray}
\label{eq:michelangelo}
 x_k&:=&H_{00;01}+\sum_{r=1}^R \cos \left(\upsilon_{k,r} \right)(H_{00;r1}-H_{01;r0})\\
 \label{eq:donatello}
 a_k&:=&-2\sum_{r=1}^R \sin \left(\upsilon_{k,r} \right)H_{00;r0}\\
 \label{eq:raphael}
 b_k&:=&-\sum_{r=1}^R \sin \left(\upsilon_{k,r} \right)(H_{00;r1}+H_{01;r0})\\
 \label{eq:leonardo}
 c_k&:=&-2\sum_{r=1}^R \sin \left(\upsilon_{k,r} \right)H_{01;r1},
\end{eqnarray}
with $\upsilon_{k,r}:= r\pi (2k - (p+1)/2)/n$.
Depending on the parity of $n$, the following eigenvalues also appear:
\begin{equation}
\varepsilon_{0,\pm} := \sum_{q=0}^{n-1} (\pm 1)^q H_{00;q1}.
\label{eq:estellico}
\end{equation}
If $p=-1$, then $\varepsilon_{0,+}$ always appears and $\varepsilon_{0,-}$ only appears if $n$ is even. If $p=1$, then $\varepsilon_{0,-}$ only appears if $n$ is odd and $\varepsilon_{0,+}$ does not appear (see Appendix \ref{APP:QuantumTI}).

The superselection rule (\ref{eq:superselectionrule}) to be fulfilled in this case reads
\begin{equation}
p = (-1)^{\left\lfloor\frac{n+(p-1)/2}{2}\right\rfloor}\prod_{k=0}^{n-1}s_k.
\label{eq:megahypersuperselectionrule}
\end{equation}
Note that if we just want to find the ground state energy of a TI fermionic Hamiltonian (\ref{eq:QuadraticAs}), then the matrices $A$ and $B$ are circulant, which means that the previous analysis can be done assuming that $p=-1$ and no superselection rule needs to be obeyed in this case, as Eq. (\ref{eq:megahypersuperselectionrule}) appears from the transformation of spins to fermions.
This result is applied to the example of Section \ref{ex:8}.

\section{Examples}
\label{sec:examples}

In this section we present three different examples in which we illustrate the tools we have presented through the paper. In Example \ref{ex:8} we optimize a tight TI inequality for $n=8$ parties with $R=2$ and PBC, showing that it has quantum violation when the same set of qubit measurements are performed at each site. Interestingly, such optimal measurements are ${\cal M}_0 = \sigma_x$, ${\cal M}_1 = \sigma_y$ and ${\cal M}_2 = \sigma_z$. In Example \ref{ex:FlipFlop} we construct a quasi TI Bell inequality for any even number of parties and any number of measurements, which depends on one parameter. We find its classical bound analytically with dynamic programming. The Bell operator corresponds to a spin XY model which we also solve analytically. Finally, in Example \ref{ex:SpinGlass} we show that the ground state of a spin glass is nonlocal in some parameter region.

\subsection{A translationally invariant Bell inequality for 8 parties}
\label{ex:8}
The general form of a translationally invariant Bell inequality $I + \beta_C \geq 0$ with $m=d=R=2$ is (cf. Eq. (\ref{eq:BellIneqGeneral}))
\begin{equation}
 I:= {\gamma} {\cal T}_2 + \sum_{k,l\in \{0,1\}}\left({\gamma}_{k,l} {\cal T}_{k,l} + \gamma_{k,2,l} {\cal T}_{k,2,l}\right),
 \label{eq:BIRis2}
\end{equation}
where the translationally invariant correlators $\cal T$ are defined as
\begin{equation}
 {\cal T}_{k_1, \ldots, k_r}:=\sum_{i=0}^{n-1} M_{(k_1,\ldots, k_r)}^{(i,r)}.
\end{equation}
In Appendix \ref{APP:TIIneq} we present a table with the optimal (tight) Bell inequalities of these kind for $n\leq 8$ and the quantum violation we can observe. Let us remark that finding all Bell inequalities for a given scenario is a computationally very expensive task and one typically manages to do it only for very small values of $n$, $m$, $d$ and $R$, even if symmetries are imposed \cite{TI50years}. When looking for Bell inequalities of the form (\ref{eq:BIRis2}), $n = 8$ was the maximum number of parties for which this task could be carried out in a reasonable time.
Here we present a particular case as an example.

If one takes the following coefficients: $\gamma = 0$, $\gamma_{00}=\gamma_{10}=-\gamma_{01}=-\gamma_{11}=2$, $-\gamma_{020}=-\gamma_{021}=\gamma_{120}=\gamma_{121}=1$, then the dynamic programming gives $\beta_C=32$ and the measurement settings ${\cal M}_0 = \sigma_x$, ${\cal M}_1 = \sigma_y$, ${\cal M}_2=\sigma_z$ produce a quantum violation of $QV =  \langle I\rangle  + \beta_C \simeq -0.2187$.

The latter is proven by applying Eq. (\ref{eq:WilliamsonTI}) to the example. More specifically, we observe that the chosen coefficients and measurements yield an $H$ matrix (cf. Eq. (\ref{eq:QuadraticHamiltonian}) and Appendix \ref{APP:QuantumTI}) with upper-diagonal blocks $h_1$ and $h_2$
\begin{equation}
h_1=\left(
\begin{array}{cc}
2&2\\2&2
\end{array}
\right), \qquad
h_2=\left(
\begin{array}{rr}
1&-1\\-1&1
\end{array}
\right),
\end{equation}
and the rest of the $h_r$ blocks are zero. This greatly simplifies the expression for $\varepsilon_k$ as $x_k = 0$ and $a_k=c_k$ imply $\varepsilon_k = 2 (a_k \pm |b_k|)$ (Note, that in the range of interest, $b_k \leq 0$ so that $\varepsilon_k = 2 (a_k \mp b_k)$).
If we take the plus sign, we have
\begin{equation}
\varepsilon_{k,+}=-8 \cos\left(\pi \frac{4k+3-p}{8}\right),
\end{equation}
and if we take the minus sign we obtain
\begin{equation}
\varepsilon_{k,-}=16 \cos\left(\pi \frac{4k+11-p}{16}\right).
\end{equation}

Thus, we can now calculate the ground state energy consistent with each $p$, which is given by
\begin{equation}
E_0 = -16-8\sqrt{2} \nonumber
\end{equation}
if $p=-1$, and
\begin{equation}
E_0 = -8\left(\sqrt{2}+2\cos(\pi/8) + 2 \sin (\pi/8)\right) \nonumber
\end{equation}
if $p=1$.
One does not need to check the superselection rule Eq. (\ref{eq:megahypersuperselectionrule}) for $p=-1$ as some of the $\varepsilon_k$ are zero. However, we do need to check it for $p=1$. There, we took an even ($2$) number of sign flips to the $\varepsilon_k$ and the determinant of $O$ is $1$ (cf. Appendix \ref{APP:QuantumTI}) . It follows that Eq. (\ref{eq:superselectionrule}) holds. The case for $p=-1$ gives $E_0 \simeq -27.3137$ whereas for $p=1$ it gives $E_0 = -32.2187$. Hence, $\langle I \rangle + \beta_C = -0.2187 < 0$, signalling the presence of nonlocality.

\subsection{A quasi translationally invariant Bell inequality}
\label{ex:FlipFlop}
Let us consider the chained Bell inequality \cite{ChainedInequality} between two parties labelled $A$ and $B$:
\begin{equation}
  I_{\mathrm{chain}}^{(A, B)} \geq -2(m-1),
 \label{eq:CHAIN}
\end{equation}
where $I_{\mathrm{chain}}^{(A,B)}$ is given by 
\begin{equation}
I_{\mathrm{chain}}^{(A,B)}:=\sum_{i=0}^{m-1}\left(A_{m-i-2}B_i+ A_{m-i-1}B_i\right),
\label{eq:def:CHAIN}
\end{equation}
where it is assumed that $A_{-1}:=-A_{m-1}$.  Note that the CHSH inequality \cite{CHSH} $I_{CHSH}^{(A, B)}:= A_0 B_0 + A_0 B_1 + A_1 B_0 - A_1 B_1$ is a particular case of Eq. (\ref{eq:def:CHAIN}) for $m=2$.
Inequality (\ref{eq:CHAIN}) is violated maximally with the following settings
\begin{equation}\label{Ai}
A_i=\sin(\phi_i) \sigma_x-\cos(\phi_i) \sigma_y
\end{equation}
and $B_i=A_i$, where the angles are given by
$\phi_i:=(i+1)\pi/m$, and with the state 
\begin{equation}
\ket{\psi_m} := \frac{1}{\sqrt{2}}\left(e^{-\frac{\mathbbm{i}\pi}{2m}}\ket{00} - \ket{11}\right),
\end{equation}
giving $\beta_Q=\langle I_{\mathrm{chain}}^{(A,B)}\rangle =-2m\cos(\pi/2m)$. Notice that 
the maximal violation relative to the classical bound is $\beta_Q^{r}:=\beta_Q/\beta_C=[m/(m-1)]\cos(\pi/2m)$.

The bipartite Bell operator corresponding to the chained Bell inequality with the above 
measurements can be written as
\begin{equation}
 {\cal B} = \alpha_m (\sigma_x\ot \sigma_x-\sigma_y\ot\sigma_y)+\beta_m(\sigma_x\ot\sigma_y+\sigma_y\ot\sigma_x)
 \label{eq:exemple2BCBellOp}
\end{equation}
where $\alpha_m:=m\cos^2(\pi/2m)$ and $\beta_m:=(m/2)\sin(\pi/m)$. By defining 
$\sigma_{\pi/2m}:=\cos(\pi/2m)\sigma_x+\sin(\pi/2m)\sigma_y$, this 
operator can be further re-expressed in a formally similar manner to the XY Hamiltonian as 
\begin{equation}
 {\cal B}= m \left(\sigma_{\pi/2m}^{(A)}\sigma_{\pi/2m}^{(B)} - \sigma_y^{(A)}\sigma_y^{(B)}\right).
\end{equation}

Let us now consider the following Hamiltonian:
\begin{equation}
 {\cal H}=m\sum_{i=0}^{2n-1}  f_i(\epsilon)\left(\sigma_{\pi/2m}^{(i)}\sigma_{\pi/2m}^{(i+1)} - \sigma_y^{(i)}\sigma_y^{(i+1)}\right),
 \label{eq:Hamiltonian:FlipFlop:BC}
\end{equation}
where the weights $f_i(\epsilon)$ alternate from even to odd sites as $f_i(\epsilon):=1+(-1)^i \epsilon$ and $\epsilon$ is an arbitrary real parameter.
We note that the Hamiltonian (\ref{eq:Hamiltonian:FlipFlop:BC}) is a particular case of the one-dimensional Bell inequality
\begin{equation}
 I_{\mathrm{chain}}^{(2n)}(\epsilon) := \sum_{i=0}^{2n-1} f_i(\epsilon) I_{\mathrm{chain}}^{(i,i+1)}
\end{equation}
when the same measurements (\ref{Ai}) are taken at each site.

Let us now determine the classical bound of 
$ I_{\mathrm{chain}}^{(2n)}(\epsilon)$. For even\footnote{For odd $n$ and $n < m$, the classical bound is slightly different. However, for the purposes of the present example, it is enough to consider the classical bound for even $n$.} $n$, the dynamic programming gives 
$\beta_C = 4n(m-1) \max\{1, |\epsilon|\}$. The explanation for this result is that, whenever $\varepsilon > 1$, it is better to use a classical strategy that would give 
$-2(m-1)$ on every $I_{\mathrm{chain}}^{(2i, 2i+1)}$ inequality and $2(m-1)$ on every $I_{\mathrm{chain}}^{(2i+1, 2(i+1))}$ inequality. An exemplary local strategy achieving this bound is given by
\begin{equation}
\begin{array}{c|cccc}
k&4i&4i+1&4i+2&4i+3\\
\hline
M_0^{(k)}&+&-&+&+\\
M_1^{(k)}&+&-&-&+\\
\vdots&\vdots&\vdots&\vdots&\vdots\\
M_{m-2}^{(k)}&+&-&-&+ \\
M_{m-1}^{(k)}  &+&+&-&- 
\end{array},\nonumber
\end{equation}
periodically repeated every $4$ sites.
On the other hand, if $0 \leq \epsilon \leq 1$, the optimal strategy consists in producing $-2(m-1)$ for every link; for instance, by picking 
$(M_0^{(2i)}, \ldots, M_{m-1}^{(2i)}) = (+,\ldots,+)$ and
$(M_0^{(2i+1)}, \ldots, M_{m-1}^{(2i+1)}) = (-,\ldots,-)$. 
The analysis for $\epsilon < 0$ is analogous.

It is worth highlighting two limiting cases:
\begin{itemize}
 \item $\epsilon = 1$. This corresponds to a sum of disjoint chained Bell inequalities between pairs $(2i, 2i+1)$.
 \begin{equation}
  I^{(2n)}_{\mathrm{chain}}(1)=2\sum_{i=0}^{n-1} I_{\mathrm{chain}}^{(2i, 2i+1)},
 \end{equation}
which is maximally violated by the state $\ket{\psi_m}_{AB}\otimes \ket{\psi_m}_{CD}\otimes \cdots$. 
The quantum value is then $\beta_Q = 2nm\cos(\pi/2m)$. Hence, there is a $O(1)$ violation relative to the classical 
bound that holds for every $n$:
\begin{equation}
 \beta_Q/\beta_C = \frac{m}{m-1}\cos\left(\frac{\pi}{2m}\right) > 1.
 \label{eq:RatioBCEps1}
\end{equation}

\item $\epsilon = 0$. This case corresponds to a sum of the chained Bell 
inequalities with the same weights between neighbours:
\begin{equation}
I^{(2n)}_{\mathrm{chain}}(0)=\sum_{i=0}^{2n-1}I^{(i,i+1)}_{\mathrm{chain}}.
\end{equation}
This inequality cannot be violated, as quantum correlations are monogamous with respect to the chained Bell inequality \cite{TonerVerstraete, RaviHorodecki}. Loosely speaking, if party $B$ violates $I_{\mathrm{chain}}$ with $A$, it cannot violate it simultaneously with $C$. For some types of monogamy relations, this result holds for various generalizations of the CHSH inequality to more measurements, outcomes and parties \cite{Monogamies}.
\end{itemize}
It is then clear that there is some critical value of $\epsilon$ for which correlations stop being nonlocal and one is able to simulate them locally.

Let us now notice that the $4n\times 4n$ matrix $H$ appearing in Eq. (\ref{eq:QuadraticHamiltonian}) and corresponding to the 
Hamiltonian (\ref{eq:Hamiltonian:FlipFlop:BC}) has the form $H=H_0 \otimes H_1$, where
\begin{equation}
H_0:=m\left(
\begin{array}{cccccc}
0&f_0&0&\cdots & 0 & p f_1\\
-f_0 & 0 & f_1 &0 & \cdots & 0\\
0 & -f_1 & 0 & f_0 & 0 &\\
&&\ddots & \ddots & \ddots &\\
0& \cdots & 0 & -f_1 & 0 &f_0\\
-p f_1 & 0 & \cdots & 0 & -f_0 & 0
\end{array}
\right),
\end{equation}
with $f_0 :=1+\epsilon$ and $f_1 = 1 - \epsilon$ for short, and 
\begin{eqnarray}
H_1&=&\left(
\begin{array}{cc}
\frac{1}{2}\sin(\frac{\pi}{m})& \cos^2(\frac{\pi}{2m})\\[1ex]
\cos^2(\frac{\pi}{2m}) &-\frac{1}{2}\sin(\frac{\pi}{m})
\end{array}
\right)\nonumber\\
&=&\alpha_m\sigma_x+\beta_m\sigma_z.
\end{eqnarray}
This is seen by applying the JW transformation to Eq. (\ref{eq:exemple2BCBellOp}).

In this case, to find the Williamson eigenvalues of $H$ it is sufficient to find those of $H_0$, which will appear with both signs, as $H_1$ has eigenvalues $\pm \cos(\pi/2m)$.

A similar analysis as the previous example shows that $H$ can also be block-diagonalized with a real DFT into $4\times 4$ blocks. The Williamson eigenvalues of $H_0$ are
\begin{equation}
\varepsilon_k = m \sqrt{2 \left[1 + \epsilon^2 + (\epsilon^2-1)\cos(\nu_{k})\right]}\quad 0 \leq k < n,
\end{equation}
where $\nu_k := \pi (2k+(p+1)/2)/n$. 

Hence, the quantum bound will be
\begin{equation}
\beta_Q(\epsilon) = 2\cos\left(\frac{\pi}{2m}\right)\sum_{k=0}^{n-1}\varepsilon_k.
\end{equation}
Although one should check the superselection rule Eq. (\ref{eq:superselectionrule}), it is not necessary for large $n$, as the difference between $\beta_Q$ for $p=1$ and $\beta_Q$ for $p=-1$ vanishes as $n$ grows.

Let us analyze the behaviour of $\beta_Q(\epsilon)$ in the thermodynamic limit. 
The contribution per particle to $\beta_Q(\epsilon)$, denoted $\widetilde{\beta}_Q(\epsilon)$ is
\begin{equation}
\widetilde{\beta}_Q(\epsilon):=2\cos\left(\frac{\pi}{2m}\right)\lim_{n\rightarrow \infty} 
\sum_{k=0}^{n-1}\frac{1}{2n}\varepsilon_k, \nonumber
\end{equation}
which is a Riemann sum, so it is by definition
\begin{eqnarray}
\widetilde{\beta}_Q(\epsilon) &=&\sqrt{2}m\cos\left(\frac{\pi}{2m}\right)\nonumber\\
&&\times \int_{0}^{1}\sqrt{1+\epsilon^2+(\epsilon^2-1)\cos(2\pi x)}\,\mathrm{d}x.
\end{eqnarray}
This can be expressed more compactly as
\begin{equation}
\widetilde{\beta}_Q(\epsilon) = \frac{4}{\pi}m\cos\left(\frac{\pi}{2m}\right) E(1-\epsilon^2),
\label{eq:EffectiveQuantum}
\end{equation}
where $E(t)$ is the complete elliptic integral of the second kind\footnote{The complete Elliptic integral of the second kind is defined as $$E(t):=\int_{0}^{\pi/2}\sqrt{1-t\sin^2(\theta)}\mathrm{d}\theta,$$ with the parameter $t$ obeying $0<t<1$.}.

In Figure \ref{fig:CHSH-chain} we can see such behavior compared to the contribution per particle to the classical bound, $\tilde{\beta}_C(\epsilon) = 2(m-1) \max \{1,|\epsilon|\}$.

\begin{center}
\begin{figure}
 \includegraphics[scale=.5]{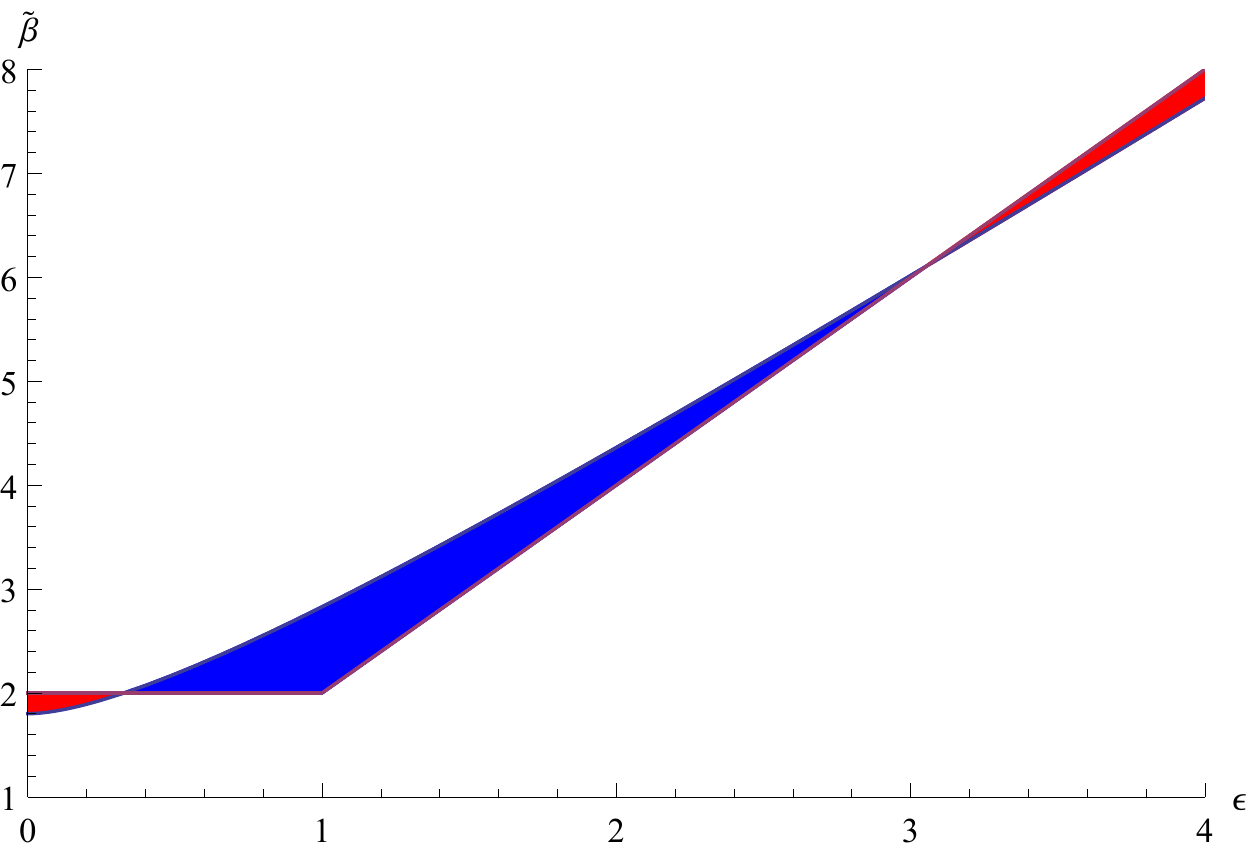}
\caption{For $m=2$, the curves $\tilde{\beta}_Q(\epsilon)$ (cf. Eq. (\ref{eq:EffectiveQuantum})) and $\tilde{\beta}_C(\epsilon)=2\max\{1,|\epsilon|\}$. These capture the behaviour of the nonlocality of the ground state of (\ref{eq:Hamiltonian:FlipFlop:BC}) in the limit of large $n$. Whenever $\tilde{\beta}_Q(\epsilon)$ is above $\tilde{\beta}_C(\epsilon)$ (blue region), nonlocality is detected. Otherwise (red region) a more stringent test is needed or the state admits a local hidden variable model. The intersection points between the two curves are $\epsilon_* \approx 0.327618$ and $\epsilon^* \approx 3.05234$.}
\label{fig:CHSH-chain}
\end{figure}
\end{center}

\begin{center}
\begin{figure}
 \includegraphics[scale=.5]{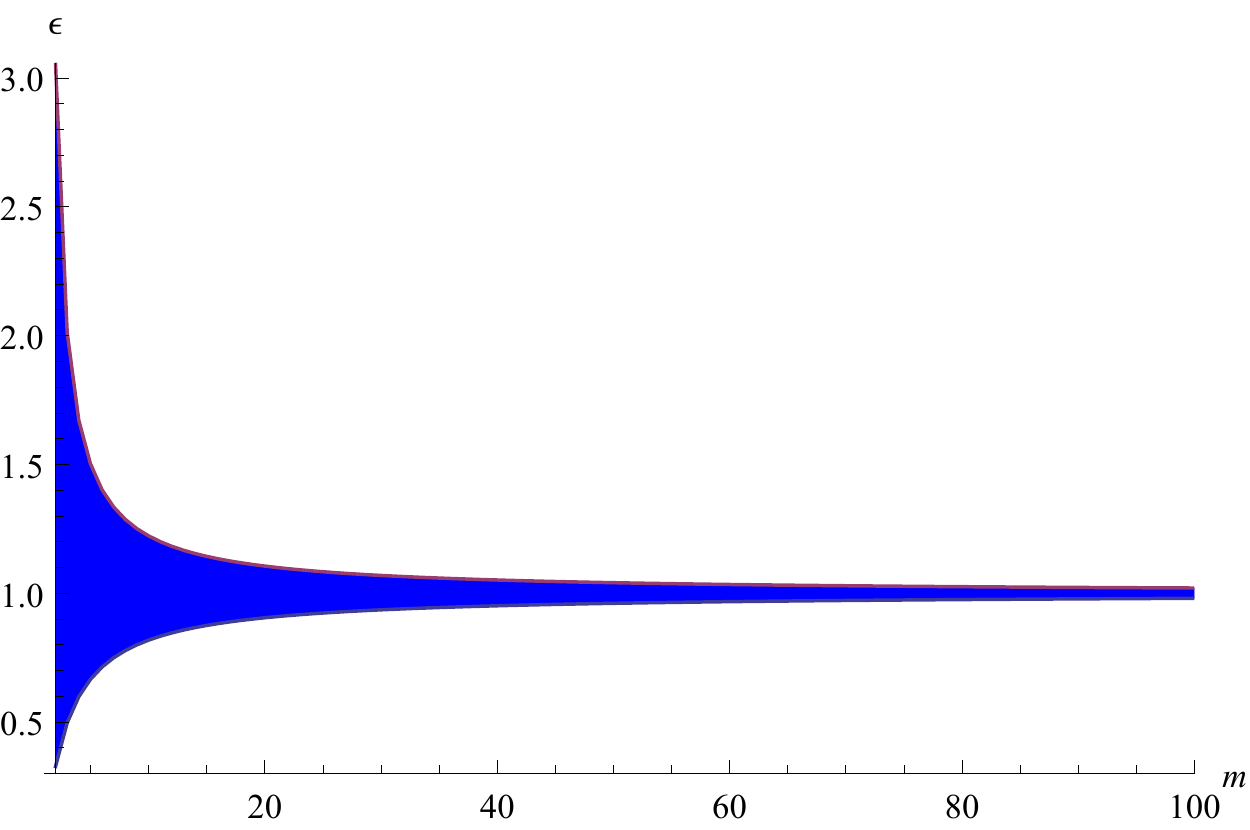}
\caption{For different values of $m$, the intervals $[\epsilon_*,\epsilon^*]$ for which we detect that the ground state of Eq. (\ref{eq:Hamiltonian:FlipFlop:BC}) is nonlocal. Observe that, for $m=2$, which corresponds to the case where the chained Bell inequality is the CHSH Bell inequality, we detect nonlocality in the largest parameter region (cf. Fig. \ref{fig:CHSH-chain}). For $\epsilon=1$ nonlocality is always detected, since $\tilde{\beta}_Q(1)/\tilde{\beta}_C(1) = m \cos (\pi/2m)/(m-1) > 1$ for all $m>1$, as it is shown in Eq. (\ref{eq:RatioBCEps1}).}
\label{fig:BC-chain}
\end{figure}
\end{center}

\subsection{A spin glass}
\label{ex:SpinGlass}
Finally we consider a Hamiltonian similar to Eq. (\ref{eq:Hamiltonian:FlipFlop:BC}) for $m=2$, where all the couplings are random, generated from a Gaussian probability distribution with mean $\mu$ and standard deviation $\sigma$, like a spin glass:
\begin{equation}
H=\sum_{i=0}^{n-1} J^{(i)}_{\mu, \sigma} \left(\sigma_{\pi/4}^{(i)} \sigma_{\pi/4}^{(i+1)}-\sigma_y^{(i)} \sigma_{y}^{(i+1)}\right).
\label{eq:SpinGlass}
\end{equation}

We can efficiently compute the ground state of Eq. (\ref{eq:SpinGlass}) and the classical bound of the CHSH-like Bell inequality associated to it with the methods we have presented, although numerically. We expect that, when $|\mu|/\sigma \gg 1$, we do not detect nonlocality, as the resulting inequality is close to the monogamous limit $\epsilon=0$ of Example \ref{ex:FlipFlop}. However, if $\sigma$ is big enough, one expects to have links with high weight surrounded with links with smaller weight; then, it may compensate to violate those links with higher weight.
In Figure \ref{fig:Spin-Glass} we show the result, for $n=10^2$ spins, averaged over $10^3$ realizations. There is clearly a transition for which the complexity of the ground state allows or does not allow for the statistics that one obtains when performing measurements on it to be simulable locally.

\begin{center}
\begin{figure}
 \includegraphics[scale=.5]{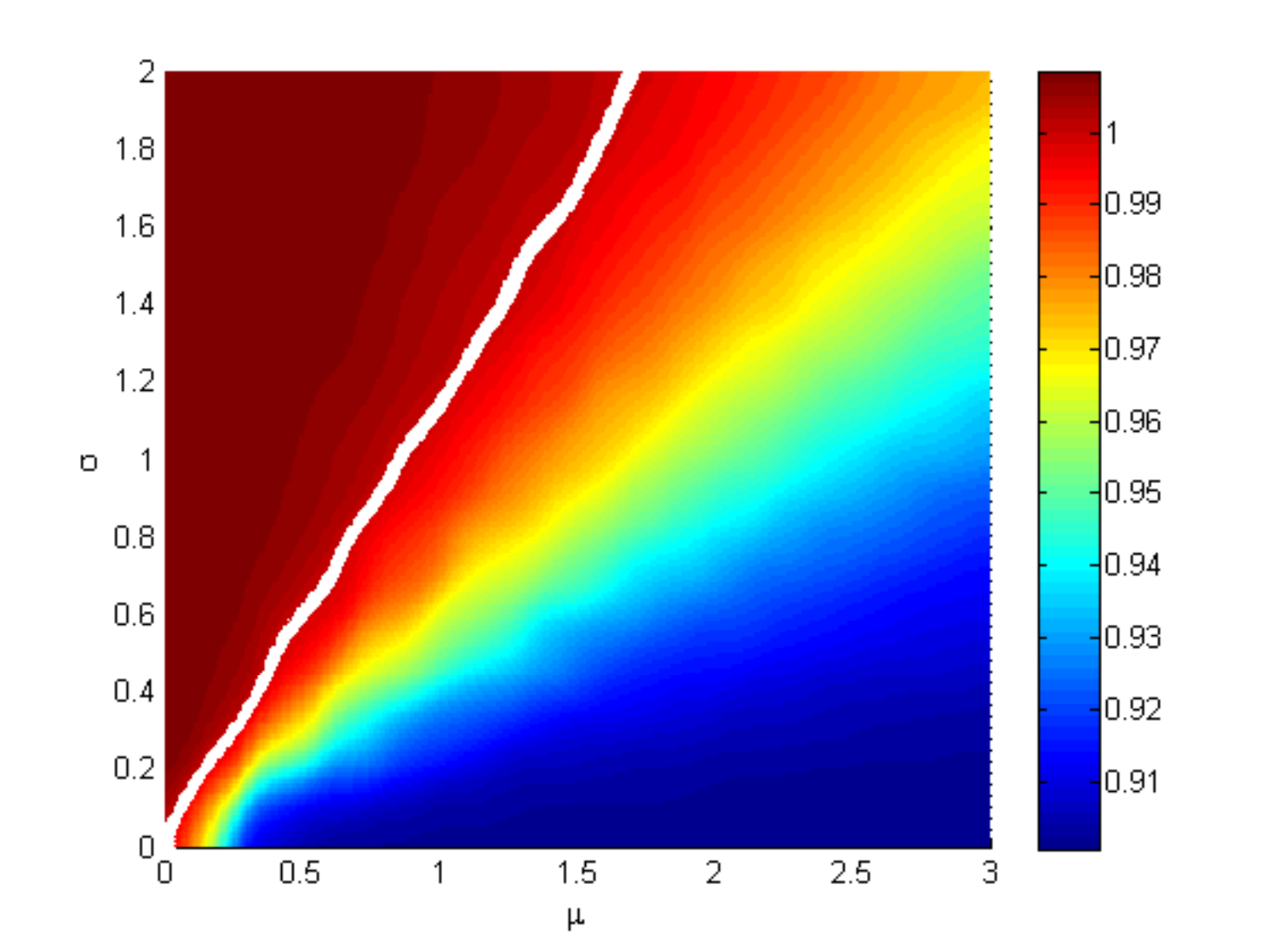}
\caption{Ratio of the quantum value over the classical bound of the ground state of Eq. (\ref{eq:SpinGlass}) and the Bell inequality we associate to it. The plot corresponds to a spin glass of $n=100$ spins with PBC, averaged over $1000$ realizations. The horizontal axis corresponds to the mean $\mu$ of the Gaussian distribution and the vertical axis corresponds to its standard deviation $\sigma$. If $\mu = 0$, the value is constant for all $\sigma > 0$, as both the expected value of the ground state and the classical bound grow at the same rate with $\sigma$. The white line represents the level curve for $\beta_C = \beta_Q$. The top-left part of the plot corresponds to the region of parameters in which we detect nonlocality ($\beta_C < \beta_Q$). Note that on the bottom-right region one finds values for which $\beta_Q/\beta_C < 1$ due to the fact that there is no simultaneous eigenvalue of $\sigma_{\pi/4}$ and $\sigma_y$ (the classical bound cannot be saturated using $\sigma_{\pi/4}$ and $\sigma_y$ as observables).}
\label{fig:Spin-Glass}
\end{figure}
\end{center}

\subsection{An XXZ-like spin model based on Gisin's Elegant Bell inequality}
\label{ex:GisinElegant}
In this last example, we present an XXZ-like Hamiltonian, which is not solvable via the JW transformation. We find its ground state energy numerically, using tensor networks and DMRG \cite{iTensor}. In this case we find a much richer structure than in Example \ref{ex:FlipFlop}. The Bell inequality that we associate to this Hamiltonian is a modification of Gisin's Elegant Bell inequality \cite{GisinElegant}. Gisin's original inequality is defined in a bipartite scenario with four dichotomic measurements with outcomes $\pm 1$ on Alice and three dichotomic measurements with outcomes $\pm 1$ on Bob:
\begin{equation}
I = \left(
  \begin{array}{cccc}
   A_0&A_1&A_2&A_3
  \end{array}
  \right) \left(
  \begin{array}{rrr}
   1&1&1\\
   1&-1&-1\\
   -1&1&-1\\
   -1&-1&1
  \end{array}\right) \left(
  \begin{array}{c}
   B_0\\B_1\\B_2
  \end{array}\right),
 \label{eq:GisinElegant}
\end{equation}
and it reads $|I| \leq 6$.
We observe that with the following observables:
 \begin{equation}
  A_0 = \frac{\sigma_x + \sigma_y + \sigma_z}{\sqrt{3}}, \quad A_1 = \frac{\sigma_x - \sigma_y - \sigma_z}{\sqrt{3}},\nonumber
 \end{equation}
  \begin{equation}
  A_2 = \frac{-\sigma_x + \sigma_y - \sigma_z}{\sqrt{3}}, \quad A_3 = \frac{-\sigma_x - \sigma_y + \sigma_z}{\sqrt{3}},\nonumber
 \end{equation}
 \begin{equation}
  B_0 = \sigma_x, \quad B_1 = \sigma_y, \quad B_2 = \sigma_z,
 \end{equation}
 the corresponding Bell operator becomes
 \begin{equation}
  {\cal B} = \frac{4}{\sqrt{3}}\left(\sigma_x \sigma_x + \sigma_y \sigma_y + \sigma_z \sigma_z\right),
 \end{equation}
 and its ground state is $\ket{\psi^-}=(\ket{01}-\ket{10})/\sqrt{2}$, yielding an expectation value $\bra{\psi^-}{\cal B}\ket{\psi^-} = -4\sqrt{3} \simeq -6.9282$.

 Let us now introduce the following modification, where $\Delta$ is a real parameter,
 \begin{equation}
  J_{\mathrm{even}} = \left(
  \begin{array}{cccc}
   A_0&A_1&A_2&A_3
  \end{array}
  \right) S_\Delta \left(
  \begin{array}{c}
   B_0\\B_1\\B_2
  \end{array}\right),
 \end{equation}
 \begin{equation}
  J_{\mathrm{odd}} = \left(
  \begin{array}{cccc}
   B_0&B_1&B_2
  \end{array}
  \right) S_\Delta^T \left(
  \begin{array}{c}
   A_0\\A_1\\A_2\\A_3
  \end{array}\right),
 \end{equation}
 where $S_\Delta$ is a $4\times 3$ matrix defined as
 \begin{equation}
  S_\Delta = \left(
  \begin{array}{rrr}
   1&1&\Delta\\
   1&-1&-\Delta\\
   -1&1&-\Delta\\
   -1&-1&\Delta
  \end{array}\right).
 \end{equation}
 The classical bound becomes
 \begin{equation}
  J_{\mathrm{even/odd}} \geq \left\{
  \begin{array}{rrr}
   4\Delta & \mbox{if} & \Delta \leq -2\\
   -4+2\Delta & \mbox{if} & -2 < \Delta \leq 0\\
   -4-2\Delta & \mbox{if} & 0 < \Delta \leq 2\\
   -4\Delta & \mbox{if} & 2 < \Delta
  \end{array}.
  \right.
 \end{equation}
 Now we have that the Bell operator has become
 \begin{equation}
  {\cal B} = \frac{4}{\sqrt{3}}\left(\sigma_x \sigma_x + \sigma_y \sigma_y + \Delta \sigma_z \sigma_z\right)
 \end{equation}
 in either case. Its ground state energy is
 \begin{equation}
  \bra{\psi_{\textrm{gnd}}}{\cal B}\ket{\psi_{\textrm{gnd}}} = \left\{
  \begin{array}{ccc}
   \frac{4\Delta}{\sqrt{3}} & \mbox{if} & \Delta \leq -1\\[2ex]
   \frac{-4(2+\Delta)}{\sqrt{3}} &\mbox{if} & \Delta > -1
  \end{array}.
  \right.
 \end{equation}
 The ground state is $\ket{\psi^-}$ if $\Delta > -1$ and it lies in the subspace spanned by $\ket{00}$ and $\ket{11}$ if $\Delta \leq -1$.

 We can now construct the many-body Bell inequality in a similar fashion as in Example \ref{ex:FlipFlop}:
 \begin{equation}
  J = \sum_{i=0}^{n/2-1} (1+\epsilon) J_{\mathrm{even}}^{(2i,2i+1)} + (1-\epsilon) J_{\mathrm{odd}}^{(2i+1,2i+2)}.
  \label{eq:IneqXXZ}
 \end{equation}
 Note that the Bell inequality $J$ has $4$ binary measurements with outcomes $\pm 1$ on the even sites and $3$ binary measurements with outcomes $\pm 1$ on the odd sites.

 The dynamic programming procedure yields the following classical bound in terms of $\epsilon$ and $\Delta$, which is a ${\cal C}^0$ piece-wise continous function. Due to all the cases that appear, and the complexity of the inequality, we omit the description of the local deterministic strategy. The regions are defined as follows (see Fig. \ref{fig:Regions}):
 \begin{figure}[h]
 \includegraphics[width=7cm]{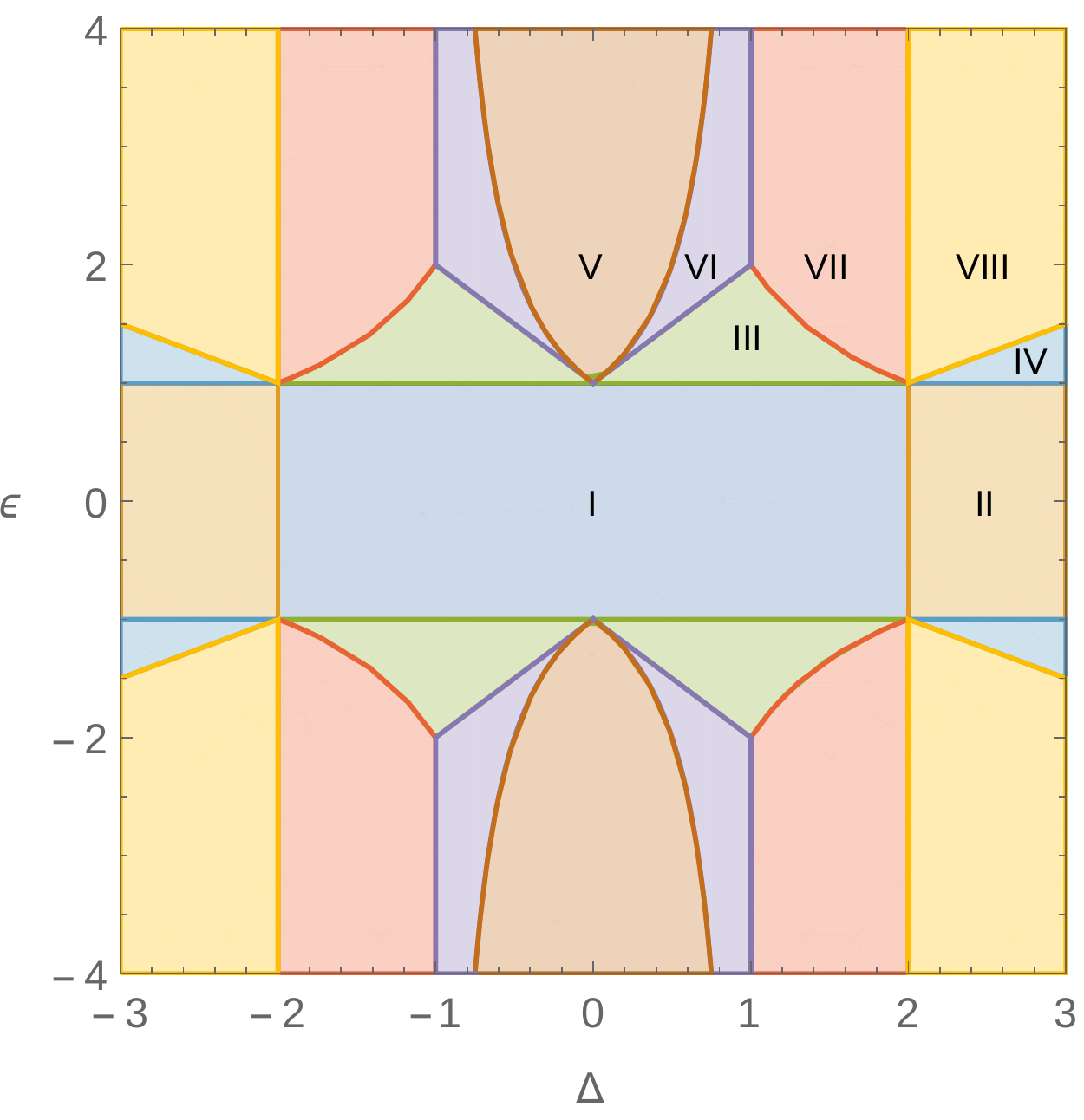}
 \caption{The classical bound of Eq. (\ref{eq:IneqXXZ}). Each region (cf. Eq. (\ref{eq:Regions})) has a different expression for the classical bound (cf. Eqs. (\ref{eq:XXZcbEven}, \ref{eq:XXZcbOdd})). The classical bound of Eq. (\ref{eq:IneqXXZ}) only depends on the absolute values of $\epsilon$ and $\Delta$.}
 \label{fig:Regions}
\end{figure}
  \begin{eqnarray}
  R_\mathrm{I}&=&\{(\Delta,\epsilon):\ |\Delta|\leq 2,\ |\epsilon| \leq 1\},\nonumber\\
  R_{\mathrm{II}}&=&\{(\Delta,\epsilon):\ |\Delta| > 2,\ |\epsilon| \leq 1\},\nonumber\\
  R_{\mathrm{III}}&=&\{(\Delta,\epsilon):\ |\Delta| \cdot |\epsilon| \leq 2,\ |\epsilon| \leq |\Delta| + 1, \ |\epsilon| > 1\},\nonumber\\
  R_{\mathrm{IV}}&=&\{(\Delta,\epsilon):\ |\epsilon| \leq  |\Delta|/2, \ |\epsilon| > 1\},\nonumber\\
  R_{\mathrm{V}}&=&\{(\Delta,\epsilon):\ |\epsilon| >  1/(1-|\Delta|)\},\nonumber\\
  R_{\mathrm{VI}}&=&\{(\Delta,\epsilon):\ |\epsilon| \leq  1/(1-|\Delta|),\ |\Delta|\leq 1, |\epsilon| > |\Delta| + 1\},\nonumber\\
  R_{\mathrm{VII}}&=&\{(\Delta,\epsilon):\ |\epsilon| \cdot |\Delta| > 2, 1 < |\Delta| \leq 2\},\nonumber\\
  R_{\mathrm{VIII}}&=&\{(\Delta,\epsilon):\ |\epsilon| > |\Delta|/2, |\Delta| > 2\}.
  \label{eq:Regions}
 \end{eqnarray}
 If $n \equiv 2 \mod 4$, $n > 2$, the classical bound is, on each region:
 \begin{eqnarray}
  \beta_{C,\mathrm{I}} &=& -n (4+2 |\Delta|),\nonumber\nonumber\\
  \beta_{C,\mathrm{II}} &=& -4  n  |\Delta|,\nonumber\nonumber\\
  \beta_{C,\mathrm{III}} &=& -8-4  |\Delta| - (4n-8)  |\epsilon| - (2n-4)  |\Delta|  |\epsilon|,\nonumber\\
  \beta_{C,\mathrm{IV}} &=& -8|\Delta| - (4n-8)|\epsilon||\Delta|,\nonumber\\
  \beta_{C,\mathrm{V}} &=& -4n|\epsilon| - (2n-8)|\epsilon||\Delta|,\nonumber\\
  \beta_{C,\mathrm{VI}} &=& -4 -(4n-4) |\epsilon| - (2n-4) |\epsilon| |\Delta|,\nonumber\\
  \beta_{C,\mathrm{VII}} &=& -4 |\Delta| - (4n-8) |\epsilon| - 2n|\epsilon||\Delta|,\nonumber\\
  \beta_{C,\mathrm{VIII}} &=& -8 |\epsilon| - 4 |\Delta| - (4n-8) |\epsilon| |\Delta|,
  \label{eq:XXZcbEven}
 \end{eqnarray}
whereas if $n \equiv 0 \mod 4$, the classical bound simplifies to
 \begin{eqnarray}
  \beta_{C,\mathrm{I}} &=& -n (4+2 |\Delta|),\nonumber\\
  \beta_{C,\mathrm{II}} &=& -4  n  |\Delta|,\nonumber\\
  \beta_{C,\mathrm{III}} = \beta_{C,\mathrm{V}} = \beta_{C,\mathrm{VI}} = \beta_{C,\mathrm{VII}}&=& -n |\epsilon| (4+2 |\Delta|),\nonumber\\
  \beta_{C,\mathrm{IV}} = \beta_{C,\mathrm{VIII}}&=& -4 n |\epsilon| |\Delta|.
  \label{eq:XXZcbOdd}
\end{eqnarray}

Using the $A_k$ measurements on the even sites and the $B_l$ ones on the odd sites, this yields the following XXZ-type Hamiltonian:
\begin{equation}
 {\cal H} = \sum_{i=0}^{n-1} \tilde{f}_i(\epsilon)\left(\sigma_x^{(i)}\sigma_x^{(i+1)}+\sigma_y^{(i)}\sigma_y^{(i+1)}+\Delta \sigma_z^{(i)}\sigma_z^{(i+1)}\right),
 \label{eq:XXZ}
\end{equation}
where $\tilde{f}_i(\epsilon) = 4 (1 + (-1)^i  \epsilon)/\sqrt{3}$.
The ground state of Eq. (\ref{eq:XXZ}) does not have an analytically closed form and has to be computed using DMRG. To do so, we have used the iTensor library \cite{iTensor}. The results are plotted in Fig. \ref{fig:DMRG}, where we observe quantum violations in a parameter region that does not seem to vanish as $n$ grows. We also observe a different behavior depending on the parity of $n/2$, which we attribute to some sort of frustration arising in the classical optimization, especially for low values of $n$.
\begin{figure}[h]
 \includegraphics[width=8cm]{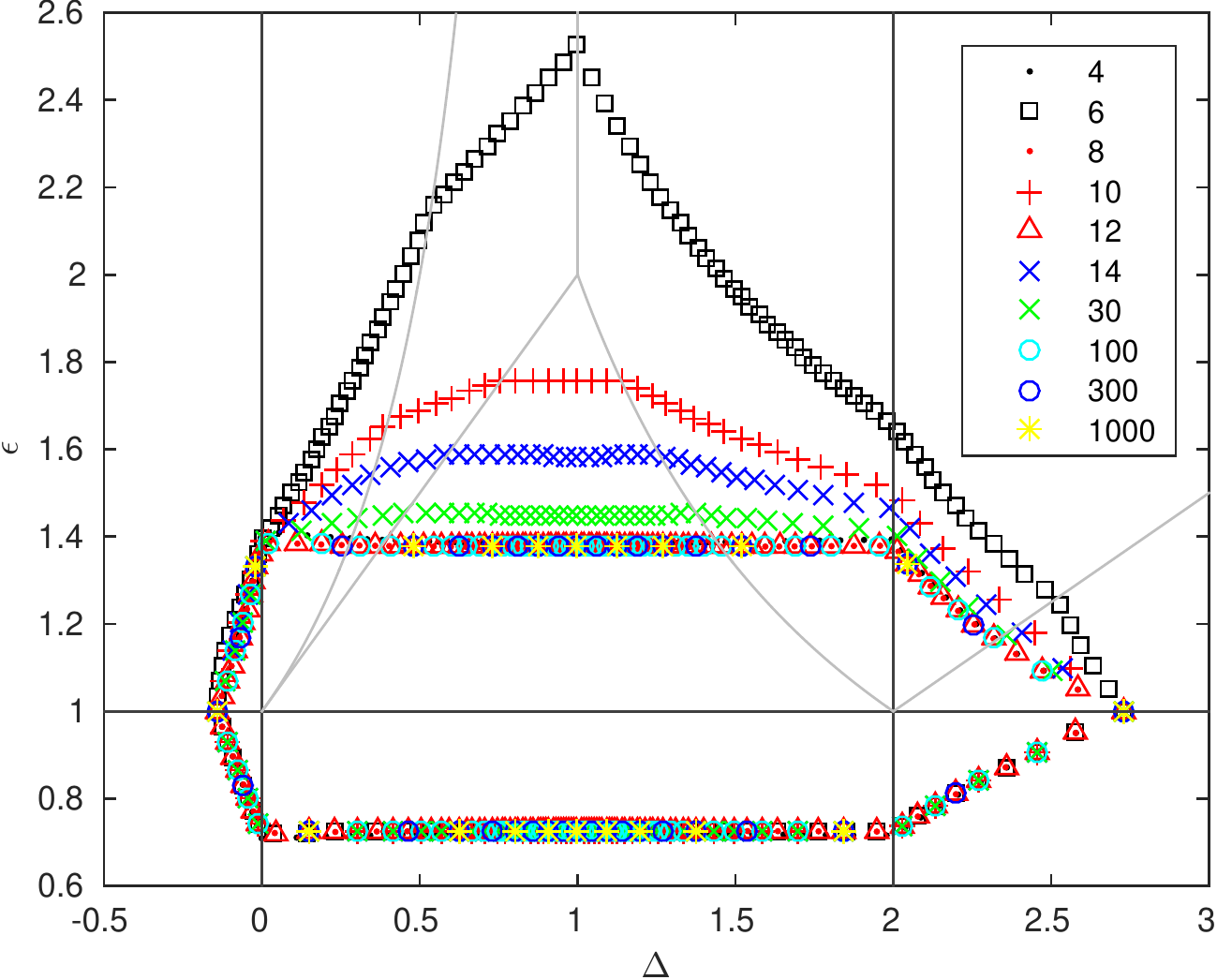}
 \caption{Level curves for which the quantum value equals the classical bound. Note the different behavior depending on the parity of $n/2$ and the effect of the different regions plotted in Fig. \ref{fig:Regions}, also sketched here in gray (for odd $n/2$) and black (even $n/2$) for clarity. The region for which there is quantum violation is the bounded set (\textit{i.e.}, containing the point $(\Delta, \epsilon) = (1,1)$). The plot is symmetric with respect to the line $\epsilon = 0$.}
 \label{fig:DMRG}
\end{figure}

\section{Conclusions and outlook}
\label{sec:conclusion}
In this work we have shown, on the one hand, that the ground states of some spin Hamiltonians are nonlocal. We have associated a Bell inequality to the Hamiltonian and we have computed its classical bound. We have exactly diagonalized the Hamiltonian and we have computed the quantum value of the corresponding Bell inequality. We have achieved this goal by combining two rather unexplored techniques in quantum information, which are dynamic programming and the Jordan-Wigner transformation.
This has allowed us to detect the presence of nonlocal correlations.
On the other hand, these tools have provided us a new way to determine the classical bound of certain classes of Bell inequalities, and look for their quantum violation under conveniently chosen observables. In the case of two dichotomic measurements, the optimization also yields the quantum bound (the maximal quantum violation) of the inequality. In the TI case, we have provided an algorithm to find the classical bound that is exponentially faster in the system size and we have obtained exact analytical solutions for the quantum value. Then, we have applied these techniques to several examples in which we reveal nonlocality: A tight TI Bell inequality with a TI Bell operator for $8$ parties, a quasi TI Bell inequality with a uniparametric quasi TI XY Hamiltonian and the ground state of a XY spin glass in some parameter region. We have also seen that for interacting models such as XXZ-like spin Hamiltonians, our method can be applied. There, we find the ground state energy numerically and we map the Hamiltonian to a modified version of Gisin's Elegant Bell inequality. These findings open the possibility to the implementation of multipartite quantum information protocols that use nonlocality as a key resource, by using ground states of Hamiltonians that appear naturally.

We remark that the Hamiltonians and the Bell inequalities we have studied have a finite interaction range, which in the TI case makes them particularly interesting from an experimentally-friendly perspective. Note that previous Bell inequalities for quantum many-body systems with low order correlators were specially designed for a permutationally invariant (PI) symmetry \cite{TASVLA, AnnPhys}; while being able to detect nonlocality in ground states of physical Hamiltonians such as the Lipkin-Meshkov-Glick \cite{LMG1, *LMG2, *LMG3} or a spin-squeezed Bose-Einstein condensate \cite{NonlocalityBEC}, the information accessible to these inequalities is bound to a de Finetti theorem \cite{deFinetti2007, deFinettiHarrow}, thus becoming more compatible with that produced by a separable state as the system grows \cite{AnnPhys}. This requires to increase the number of measurements with the system size in order to close the finite-statistics loophole \cite{NonlocalityBEC}.
However, many systems of interest are not PI, but TI, and we have studied the latter in this work. In this case, a de Finetti restriction does not apply, making the detection of nonlocality more robust to experimental imperfections.
Interestingly, to study the Bell inequalities proposed in \cite{TASVLA, AnnPhys}, in the classical and quantum many-body regime, one employs powerful mathematical tools (namely, convex hulls of semialgebraic sets \cite{ParriloSIAM} or the Schur-Weyl duality from representation theory \cite{GoodmanWallach}, respectively), which no longer apply when the PI symmetry is broken. Here we have used other mathematical tools (dynamic programming for the classical bound and the JW transformation for the quantum value) that allowed us to solve, even exactly, the two cases with this much weaker symmetry.

Let us finalize by pointing out possible future research directions that stem from our work. Throughout this paper, we have restricted ourselves to the study of nonlocality in one-dimensional spin short-ranged Hamiltonians, as we could compute their ground state energy with exact diagonalization. One can eliminate this restriction and study short-ranged spin Hamiltonians by using a Matrix Product State (MPS) description of the state, which is a good ansatz for these systems \cite{TensorNetworks}. The Bell inequalities that we would naturally associate to them would still be of the form of Eq. (\ref{eq:BellIneq}), because the interaction range $R$ would be fixed, so we could still efficiently find their classical bound. Conversely, we can eliminate the restriction on the subset of observables that we choose in Section \ref{sec:otherway} and the string of $\sigma_z$'s in the middle of the string operators, thus studying purely two-body correlator inequalities like the classes derived in \cite{TI50years}.
 In such cases, powerful numerical algorithms such as DMRG \cite{TNOrus} would be suitably tailored to perform the quantum optimization and check whether nonlocality is detected. Another interesting problem is related to the persistency of nonlocality \cite{PersistencyNL}. While the tools to carry this kind of analysis in the PI case have been put forward \cite{AnnPhys}, in the case of one spatial dimension we do not know yet how robust are the inequalities we have presented to particle losses.
One could also generalize our results towards other directions. Since Bell inequalities with more than two inputs or outputs per party are not, in general, maximally violated by measuring qubits \cite{CGLMP}, if one would like to increase the chance to detect nonlocality with these more general classes of Bell inequalities, increasing the physical dimension of the system would be a way to obtain better results. The tools presented here could be extended by using a generalized JW transformation \cite{ParafermionsHongHao} from qudits to parafermions, although the problem becomes algebraically more involved.

We have also seen through the examples presented that there is a strong relation between Hamiltonians of physical systems and Bell inequalities that we associate to them. Whereas we naturally establish this connection in the following direction: starting from the Hamiltonian of a physical system, we assign a Bell inequality to it, one can think of this relation in a more general way, since there are many Bell inequalities that, with the appropriate observables, realize the same physical Hamiltonian. Moreover, in Example \ref{ex:GisinElegant}, we have seen that this correspondence can be non-trivial, as the inequality does not even have the same number of measurements at every site. With the tools we have presented here, it is now possible that, given a physical one-dimensional Hamiltonian with short-range interactions, tailor the best Bell inequality that reveals the nonlocality of its ground state, as such Hamiltonian corresponds to the realization of a Bell operator of many different Bell inequalities, each with its own classical bound, which we can compute efficiently.

\section{Acknowledgments}

We acknowledge support of the EU Integrated Project SIQS, EU FET Proactive QUIC, Spanish MINECO (SEVERO OCHOA  programme for Centres of Excellence in R\&D SEV-2015-0522, National Plan Projects FOQUS  FIS2013-46768, FISICATEAMO FIS2016-79508-P and QIBEQI FIS2016-80773-P),  the  Generalitat  de  Catalunya  (SGR  874, SGR 875 and the CERCA programme),  Fundaci\'o Privada Cellex, ERC AdG OSYRIS and CoG QITBOX, the AXA Chair in Quantum Information Science and the John Templeton Foundation.
J. T. gratefully acknowledges a Max Planck - Prince of Asturias Award mobility grant and the CELLEX-ICFO-MPQ programme. G. D. L. C. acknowledges support from the Elise Richter fellowship of the FWF.
R. A. acknowledges funding from the European Union's Horizon 2020 research and innovation
programme under the Marie Sk\l{}odowska-Curie grant agreement No 705109.

\bibliography{biblio}

\appendix
\section{ \label{APP:parity} The parity operator}

\label{app:determinant}
Here we prove Eq. (\ref{eq:determinant}). Let $O \in {\cal O}(2n)$ be an orthogonal transformation relating the sets of Majorana fermions $\{\hat{c}_{i,\alpha}\}$ and $\{\hat{d}_{k,a}\}$, as in Eq. (\ref{eq:orthogonaltrafo}). Recall that the Cartan-Dieudonn\'e Theorem \cite{Gallier2011} states that every orthogonal transformation $O \in {\cal O}(2n)$ decomposes as a product of a number of reflections (at most, $2n$). Hence, it suffices to show that a reflection flips the parity operator; \textit{i.e.}, our aim is to show that
\begin{equation}
\prod_{i=0}^{n-1} \mathbbm{i}\hat{c}_{i,0}\hat{c}_{i,1}=-\prod_{i=0}^{n-1} \mathbbm{i}\hat{d}_{i,0}\hat{d}_{i,1},
\label{eq:app:flip}
\end{equation}
whenever $O$ is a reflection. Recall that any reflection with respect to a hyperplane with normal vector $\vec{u}$ can be written as $\mathbbm{1} - 2 \vec{u}\vec{u}^T$.
We will denote the LHS of Eq. (\ref{eq:app:flip}) $\hat{\mathbbm{P}}_c$ and the RHS of Eq. (\ref{eq:app:flip}) $-\hat{\mathbbm{P}}_d$.

The CARs (\ref{eq:CARsMajorana}) ensure that
\begin{equation}
[\mathbbm{i}\hat{c}_{i,0}\hat{c}_{i,1},\mathbbm{i}\hat{c}_{j,0}\hat{c}_{j,1}]=[\mathbbm{i}\hat{d}_{k,0}\hat{d}_{k,1},\mathbbm{i}\hat{d}_{l,0}\hat{d}_{l,1}]=0
\end{equation}
and
\begin{equation}
(\mathbbm{i}\hat{c}_{i,0}\hat{c}_{i,1})^2=(\mathbbm{i}\hat{d}_{k,0}\hat{d}_{k,1})^2=\hat{\mathbbm{1}}.
\end{equation}
Hence, the operators $\{\mathbbm{i}\hat{c}_{i,0}\hat{c}_{i,1}\}$ and $\{\mathbbm{i}\hat{d}_{k,0}\hat{d}_{k,1}\}$ respectively share an eigenbasis and split the Fock space into even and odd subspaces with respect to the parity operator $\hat{\mathbbm{P}}$ defined in Eq. (\ref{eq:def:P}).

We prove in Theorem \ref{thm:evenodd} that this partition of the Fock space does not change under orthogonal transformations. Moreover, if $\det O = -1$, the subspaces are just swapped. Theorem \ref{thm:evenodd} is supported on Lemma \ref{lem:vacuumflip}, which proves that reflections swap the eigenspaces, and Lemma \ref{lem:technical} contains the technical steps to prove Lemma \ref{lem:vacuumflip}. Hence, Eq. (\ref{eq:app:flip}) follows and, because every orthogonal transformation is a product of a number of reflections, we obtain Eq. (\ref{eq:determinant}).

\begin{lemma}
Consider a reflection $O=\mathbbm{1}-\vec{u}\vec{u}^T$, where $\vec{u}$ is a normalized vector.
Let us define the following quantities:
\begin{eqnarray}
S_0 &:=& u_{k,0}\sum_{i,\alpha}u_{i,\alpha}\hat{c}_{i,\alpha}\hat{c}_{k,1},\nonumber\\
S_1 &:=& u_{k,1}\sum_{i,\alpha}u_{i,\alpha}\hat{c}_{k,0}\hat{c}_{i,\alpha},\nonumber\\
S_{01}&:=&u_{k,0}u_{k,1}\sum_{i,\alpha,j,\beta}u_{i,\alpha}u_{j,\beta}\hat{c}_{i,\alpha}\hat{c}_{j,\beta}.
\end{eqnarray}
Then, the identities
\begin{equation}
S_{01} = u_{k,0}u_{k,1} \mathbbm{1}
\label{eq:app:S01}
\end{equation}
and
\begin{equation}
\mathbbm{i}\sum_{k=0}^{n-1}(S_0+S_1)\ket{\Omega} = \ket{\Omega} + 2\mathbbm{i}\sum_{k=0}^{n-1}u_{k,0}u_{k,1}\ket{\Omega}.
\label{eq:app:S0PlusS1}
\end{equation}
hold.
\label{lem:technical}
\end{lemma}
\begin{proof}
To prove Eq. (\ref{eq:app:S01}), we can split the sum into the indices for which $(i,\alpha)=(j,\beta)$ and the indices for which $(i,\alpha)\neq (j,\beta)$.

In the first case, since the CARs (\ref{eq:CARsMajorana}) imply that $(\hat{c}_{i,\alpha})^2 = \mathbbm{1}$ we have a term $u_{k,0}u_{k,1}\sum_{i,\alpha} u_{i,\alpha}^2(\hat{c}_{i,\alpha})^2$ which contributes $u_{k,0}u_{k,1} \mathbbm{1}$, because $\vec{u}$ is normalized.

In the second case, we note that we can rewrite the sum as
\begin{equation}
u_{k,0}u_{k,1}\sum_{(i,\alpha)<(j,\beta)}u_{i,\alpha}u_{j,\beta} \{\hat{c}_{i,\alpha},\hat{c}_{j,\beta}\} = 0,
\end{equation}
because of the CARs (\ref{eq:CARsMajorana}).

In order to prove Eq. (\ref{eq:app:S0PlusS1}), we begin by noting that
\begin{equation}
\hat{c}_{i,\alpha}\hat{c}_{k,\beta}\ket{\Omega}=\left\{
\begin{array}{lcc}
\mathbbm{i}^{\alpha + \beta}(-1)^\beta\ket{\Omega}&\mbox{ if }&i=k\\
\mathbbm{i}^{\alpha + \beta}(-1)^{\alpha + \beta}\hat{a}_i^\dagger\hat{a}_k^\dagger\ket{\Omega}&\mbox{ if }&i\neq k
\end{array}.
\label{eq:app:singleterm}
\right.
\end{equation}
We can now split the sum (\ref{eq:app:S0PlusS1}) into the parts where $i=k$ and $i\neq k$.

For the first part, we have a contribution in Eq. (\ref{eq:app:S0PlusS1}) that amounts to
\begin{equation}
\mathbbm{i}\sum_{k=0}^{n-1}\left[(u_{k,0}^2 + u_{k,1}^2)\hat{c}_{k,0}\hat{c}_{k,1} + 2 u_{k,0}u_{k,1}\mathbbm{1}\right]\ket{\Omega},\nonumber
\end{equation}
which, using Eq. (\ref{eq:app:singleterm}), simplifies to
\begin{equation}
\ket{\Omega} + 2\mathbbm{i}\sum_{k=0}^{n-1}u_{k,0}u_{k,1}\ket{\Omega}.
\end{equation}

For the second part, the contribution to Eq. (\ref{eq:app:S0PlusS1}) is
\begin{equation}
\mathbbm{i}\sum_{i\neq k}\sum_{\alpha}(u_{k,0}u_{i,\alpha}\mathbbm{i}^{\alpha+1}(-1)^{\alpha+1}-u_{k,1}u_{i,\alpha}\mathbbm{i}^{\alpha}(-1)^{\alpha})\hat{a}_i^\dagger\hat{a}_k^\dagger \ket{\Omega};\nonumber
\end{equation}
expanding the sum over $\alpha$ we have
\begin{equation}
\sum_{i\neq k}\left[\left(u_{k,0}u_{i,0}-u_{k,1}u_{i,1}\right) - \mathbbm{i}(u_{k,1}u_{i,0} + u_{k,0}u_{i,1})\right]\hat{a}_i^\dagger\hat{a}_k^\dagger \ket{\Omega}.\nonumber
\end{equation}
Splitting the sum between those indices for which $i < k$ and those for which $i > k$ we have that it can be rewritten into an expression involving the CARs (\ref{eq:CARsDirac}):
\begin{equation}
\sum_{i < k} \left[\left(u_{k,0}u_{i,0}-u_{k,1}u_{i,1}\right) -\mathbbm{i} \left(u_{k,1}u_{i,0} + u_{k,0}u_{i,1}\right)\right]\{a_i^\dagger, a_k^\dagger\}\ket{\Omega}.\nonumber
\end{equation}
Because $\{a_i^\dagger, a_k^\dagger\}=0$, this last expression is zero.
\end{proof}

\begin{lemma}
\label{lem:vacuumflip}
 Let $O$ be a reflection. Then,
 \begin{equation}
\sum_{k=0}^{n-1}\frac{\mathbbm{i}\hat{d}_{k,0}\hat{d}_{k,1}+\mathbbm{1}}{2}\ket{\Omega} = \left(\sum_{k=0}^{n-1}\frac{\mathbbm{i}\hat{c}_{k,0}\hat{c}_{k,1}+\mathbbm{1}}{2}-\mathbbm{1}\right)\ket{\Omega}.
\end{equation}
\end{lemma}
\begin{proof}
 By hypothesis, $O_{i,\alpha; j,\beta}=\delta_{i,j}\delta_{\alpha,\beta}-2u_{i,\alpha}u_{j,\beta}$, where $\delta$ is the Kronecker delta function. Let us now see how the operator $\sum_{k=0}^{n-1}\mathbbm{i}\hat{d}_{k,0}\hat{d}_{k,1}$ relates to the $\{\hat{c}_{i,\alpha}\}$ Majorana fermions.

By definition (cf. Eq. (\ref{eq:orthogonaltrafo})), we can write
\begin{equation}
\mathbbm{i}\hat{d}_{k,0}\hat{d}_{k,1} = \mathbbm{i}\hat{c}_{k,0}\hat{c}_{k,1} - 2\mathbbm{i} (S_{0}+S_{1})+4\mathbbm{i}S_{01}.\nonumber
\end{equation}

Lemma \ref{lem:technical} allows us to conclude
\begin{equation}
\sum_{k=0}^{n-1}\mathbbm{i}\hat{d}_{k,0}\hat{d}_{k,1}\ket{\Omega} = \sum_{k=0}^{n-1}\mathbbm{i}\hat{c}_{k,0}\hat{c}_{k,1}\ket{\Omega} - 2 \ket{\Omega},
\end{equation}
showing that a reflexion flips the parity of the vacuum:
The occupation number operator, in terms of the $\{\hat{d}_{k,a}\}$ operators has opposite parity than the occupation number operator in terms of the $\{\hat{c}_{i,\alpha}\}$ operators:
\begin{equation}
\sum_{k=0}^{n-1}\frac{\mathbbm{i}\hat{d}_{k,0}\hat{d}_{k,1}+\mathbbm{1}}{2}\ket{\Omega} = \left(\sum_{k=0}^{n-1}\frac{\mathbbm{i}\hat{c}_{k,0}\hat{c}_{k,1}+\mathbbm{1}}{2}-\mathbbm{1}\right)\ket{\Omega}.
\end{equation}
\end{proof}

\begin{thm}
\label{thm:evenodd}
 Let $O \in {\cal O}(2n)$. Let ${\cal F}$ be the Fock space of $n$ Dirac Fermions. The operator $\hat{\mathbbm P}_c$ splits ${\cal F}$ into even and odd subspaces: ${\cal F}={\cal F}_e \oplus {\cal F}_o$. Similarly, $\hat{\mathbbm{P}}_d$ splits ${\cal F}$ as ${\cal F}={\cal F}_e'\oplus {\cal F}_{o}'$. Then, either ${\cal F}_e = {\cal F}_e'$ (if $\det O = 1$) or ${\cal F}_e = {\cal F}_o'$ (if $\det O = -1$).
\end{thm}
\begin{proof}
Recall that the expectation value of the operator $\hat{a}_i^\dagger \hat{a}_i$ is the occupation number of the $i$-th mode, denoted $n_i$. Then, a basis of the Fock space of $n$ fermionic Dirac modes can be defined as
\begin{equation}
 \ket{\mathbf{N}}:=\ket{n_0, \ldots, n_{n-1}} := \prod_{k=0}^{n-1}(\hat{a}_k^\dagger)^{n_k} \ket{\Omega},
 \label{eq:def:fockbasis}
\end{equation}
where the product is written from left to right. Note that $n_k$ can only be $0$ or $1$, due to (\ref{eq:CARsDirac}). The CARs (\ref{eq:CARsDirac}) further imply that $\hat{a}_i \ket{\mathbf{N}}=(-1)^{\sum_{j=0}^{i-1}n_j}\ket{\mathbf{N}'}$ whenever $n_i = 1$ ($\mathbf{N}'$ and $\mathbf{N}$ differ only in its $i$-th index) and $\hat{a}_i \ket{\mathbf{N}}=0$ whenever $n_i = 0$.

 Let $\ket{\psi}$ be an eigenvector of $\hat{\mathbbm{P}}_c$ of eigenvalue $1$. We can assume, without loss of generality, that $\ket{\psi}$ is a basis element $\ket{\mathbf{N}}$ of the Fock space, with $\sum_{i=0}^{n-1}n_i \equiv 0 \mod 2$.
 By noting that $\mathbbm{i}\hat{d}_{k,0}\hat{d}_{k,1}$ can be expressed as
 \begin{equation}
  \mathbbm{i}\sum_{i,\alpha,j,\beta}\mathbbm{i}^{\alpha + \beta}O_{i,\alpha;k,0}O_{j,\beta;k,1}(\hat{a}_i+(-1)^\alpha\hat{a}_i^\dagger)(\hat{a}_j+(-1)^\beta\hat{a}_j^\dagger),\nonumber
 \end{equation}
we see that the only operators which appear are products of $\hat{a}_i$ or $\hat{a}_i^\dagger$ with $\hat{a}_j$ or $\hat{a}_j^\dagger$. Hence, they either annihilate $\ket{\mathbf{N}}$ or they add $-2$, $0$ or $2$ to its particle number; always conserving its parity. This leaves us with a linear combination of vectors that live in the same subspace ${\cal F}_e$, so it remains on ${\cal F}_e$. The same argument applies to ${\cal F}_o$.

Hence, the subspaces ${\cal F}_e$ and ${\cal F}_o$ of the operator $\hat{\mathbbm{P}}_c$ are invariant under orthogonal transformations of the Majorana fermions. However, under $\hat{\mathbbm{P}}_d$, they might change its parity. Because of linearity and Eq. (\ref{eq:def:fockbasis}), it suffices to prove that $\ket{\Omega}$ is an eigenstate of $\hat{\mathbbm{P}}_d$, as it stems from Lemma \ref{lem:vacuumflip}. If its eigenvalue is $1$, then ${\cal F}_e = {\cal F}_e'$; if it is $-1$, then ${\cal F}_e={\cal F}_o'$.
\end{proof}

\section{\label{APP:PBC2OBC} The PBC to OBC reduction}
\label{app:PBC}
In this appendix we describe the reduction of the optimization problem of finding the classical bound in a Bell inequality with PBC to the optimization problem of finding the classical bound for a Bell inequality with OBC in one dimension.

For the sake of clarity, we first describe the procedure for inequalities with an arbitrary interaction range $R$, but without the chain of $M_m^{(j)}$ observables in the middle of the string operators (cf. Eq. (\ref{eq:BellIneqGeneral})), so that the inequalities strictly contain one or two-body correlators. We pick $R$ consecutive parties. Without loss of generality, we can assume them to be labelled from $0$ to $R-1$. To these parties we assign one of the $d^{m\cdot R}$ possible deterministic local strategies. Let $(i,j)$ be a pair of parties. 
There are three cases to consider:
\begin{itemize}
\item If $i$ and $j$ are between $0$ and $R-1$, any correlator between parties $i$ and $j$ has now a definite value.
\item If either $i$ or $j$, but not both are between $0$ and $R-1$, only one side of the correlator has a definite value. Because the classical bound is computed on a deterministic local strategy, the value of the correlator factorizes as the product of the local values assigned by the strategy we are considering. Hence, these correlators can be effectively moved outside of the interval $0 \ldots R-1$ by updating the one-body term of the party outside that interval.
\item In any other case, the correlator remains the same.
\end{itemize}
In Figure \ref{fig:PBC2OBC} a) and Figure \ref{fig:PBC2OBC} b) we illustrate this procedure for an example with $R=2$. The first case is illustrated in blue; the second is illustrated in red and the last one is illustrated in black.
\onecolumngrid
\begin{center}
\begin{figure}[h!]
 \includegraphics[scale=.8]{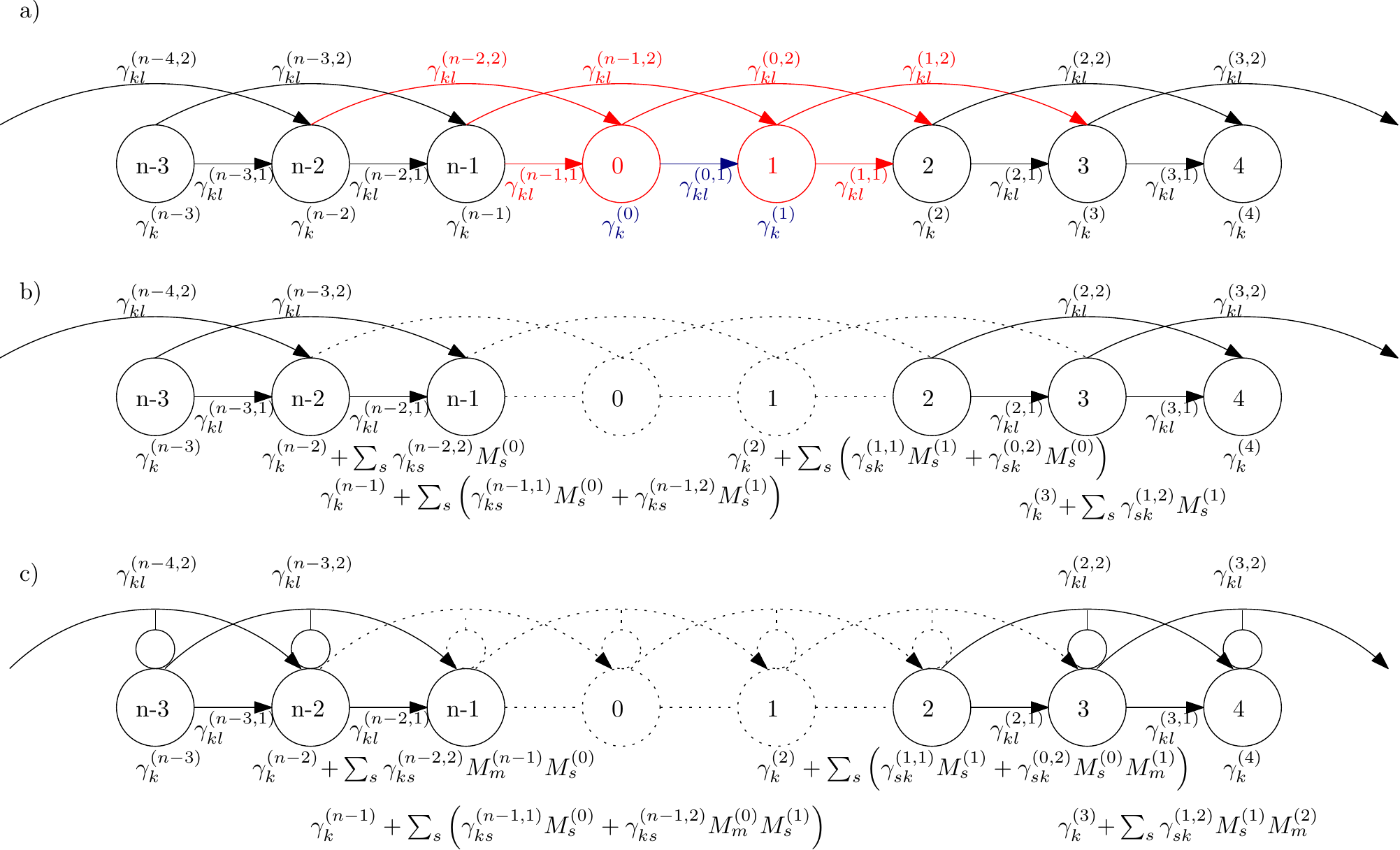}
\caption{In a), we have a two-body Bell inequality with $R=2$ and PBC. The coefficients below the circles correspond to the weights of the one-body correlators associated to those parties and the coefficients next to the arrows correspond to the weights of the correlators between the parties they join. In b), we consider the intermediate case in which there are no observables $M_m^{(i)}$ in the middle of the string operators, whereas in c) we include them (the big circles represent the set of observables labelled from $0$ to $m-1$ and the smaller circle represents the $m$-th one (cf. Eq. (\ref{eq:BellIneqGeneral}))).
In order to transform it to an OBC case, we choose a deterministic local strategy for the parties labelled by $0$ and $1$.
and we have chosen a deterministic local strategy for parties labelled by $0$ and $1$, represented in red in a). Thus, the values of the one-body correlators from $0$-th and $1$-st parties have a fixed value, as well as the two-body correlators between them. These are marked in blue. These terms contribution represents an offset on the classical bound. The two-body correlators starting at $n-2$ and ending at $0$ are updated to the one-body terms of the $(n-2)$-th party. The correlators starting at $n-1$ and ending either at $0$ or $1$ are updated to the one-body terms of the $(n-1)$-th party. Similarly, we update the one-body terms of parties $2$ and $3$. These terms are marked in red. The rest of the inequality remains untouched and such terms are marked in black. This process is the transformation from a) to b). Note that if the local deterministic strategy chosen in parties $0$ and $1$ is the optimal one, then the classical bound for the PBC problem is given by the offset generated by parties $0$ and $1$ (blue) plus the classical bound of the OBC problem between parties $2$ and $n-1$. Since there is a finite number of deterministic local strategies that parties $0$ and $1$ can have, eventually we find the optimal bound for the PBC problem. Finally, in c) we depict the same procedure, but some more assignments $M_m^{(i)}$ have to be fixed in advance to update the one-body terms accordingly.}
\label{fig:PBC2OBC}
\end{figure}
\end{center}
\twocolumngrid
Finally, if we have an inequality of the form of Eq. (\ref{eq:BellIneqGeneral}), the intermediate $M_m^{(j)}$ terms should also be taken into account when performing this procedure. In order to move to the $1$-body terms the $R$-body terms, now an extra number of $M_m^{(j)}$ observables should also be chosen in advance for those parties at distance $R-1$ to the set $\{0, \ldots, R-1\}$ (but only the $m$-th observable; it is not necessary to fix the rest). For instance, in the example of Figure \ref{fig:PBC2OBC} c), one would have to specify the value of $M_m^{(n-1)}$ and $M_m^{(2)}$. Observe that if $R>2$, then in the OBC problem, some of the $M_m^{(j)}$ values closest to the boundary become fixed ($2(R-2)$ of them). This has to be taken into account when performing the dynamical programming, as one should only explore those configurations compatible with the boundary conditions imposed by the PBC problem with the deterministic local strategy that we have chosen to make the reduction.

\section{\label{APP:QuantumTI} Quantum optimization of translationally invariant Bell inequalities}
In this appendix we discuss the quantum optimization for the translationally invariant (TI) case (\textit{i.e.}, when the Bell inequality is TI and the same set of measurements are performed at each site, thus leading to a TI Bell operator). We begin by observing that there are two cases to consider, depending on the choice of the parity $p$ of the fermion number: If $p=-1$, the matrix $H$ (cf. Eq. (\ref{eq:QuadraticHamiltonian})) is block-circulant: $H_{i,\alpha; j,\beta} = H_{i+1,\alpha; j+1,\beta}$, where the party indices are taken \textit{modulo} $n$. We can then define $h_r$ to be the $2\times 2$ block of $H$ whose entries are given by $(h_r)_{\alpha;\beta}:=H_{0,\alpha;r,\beta}$. If $p=1$, the matrix $H$ is no longer block circulant because the blocks of $H$ that correspond to interactions that cross the origin carry a minus sign (these blocks are located on the $R$ top-right diagonals and $R$ bottom-left diagonals of $H$). In this latter case, it is convenient to construct
\begin{equation}
 \tilde{H}:=\ket{-}\bra{-}\otimes H,
\end{equation}
where $\ket{-}:=(\ket{0}-\ket{1})/\sqrt{2}$, which is now circulant. Note that half of the spectrum of $\tilde{H}$ are zeroes and the other half coincides with the spectrum of $H$.
For instance, for the case of nearest neighbour interactions ($R=1$), $n=3$ parties and $p=1$, $H$ takes the block-form
\begin{equation}
H=\left(
\begin{array}{ccc}
h_0&h_1&h_1^T\\
-h_1^T&h_0&h_1\\
-h_1 & -h_1^T & h_0
\end{array}
\right).
\end{equation}
It is then clear that $\tilde{H}$ is block-circulant.

Circulant matrices can be diagonalized via a Discrete Fourier Transform (DFT). The DFT matrix is unitary in general, but we want to use an orthogonal transformation instead, so that we can transform Majorana fermions into Majorana fermions and obtain Eq. (\ref{eq:Williamson}). Let us consider the following Real DFT (RDFT) matrix of order $n$:
\begin{equation}
 ({\cal R}_n)_{kl}:=\sqrt{\frac{2}{n}}\cos\left(\frac{2\pi k l}{n} - \frac{\pi}{4}\right), \quad 0\leq k,l < n.
 \label{eq:def:RDFT}
\end{equation}
It is easy to see that ${\cal R}_n={\cal R}_n^T$ and ${\cal R}_n^2=\mathbbm{1}$, so ${\cal R}_n$ is orthogonal. In the following, we are going to use the fact that for any $n$, $\det({\cal R}_n\otimes {\mathbbm 1}_2)=(\det {\cal R}_n)^2=1$.

Let us now, with the aid of the orthogonal transformation ${\cal R}_n$, study which invariant subspaces $H$ acts upon.
If $p=-1$, then we compute $({\cal R}_n \otimes {\mathbbm{1}_2}) H ({\cal R}_n \otimes {\mathbbm{1}_2})$; if $p=1$ then we calculate $({\cal R}_{2n} \otimes {\mathbbm{1}_2}) \tilde{H} ({\cal R}_{2n} \otimes {\mathbbm{1}_2})$. Note that $\mathbbm{1}_2$ acts on the two Majorana modes asociated to one site.

On the one hand, if $p=-1$, a direct calculation of $H':=({\cal R}_n \otimes {\mathbbm{1}_2}) H ({\cal R}_n \otimes {\mathbbm{1}_2})$ shows that
\begin{equation}
 H'=\left(\bigoplus_{k=1}^{\lfloor (n-1)/2\rfloor} J_k\right) \oplus \left[\sum_{q=0}^{n-1} h_q\right] \oplus \left[\sum_{q=0}^{n-1} (-1)^qh_q\right],
 \label{eq:app:bigopluseven}
\end{equation}
where the last subspace only appears if $n$ is even and each $J_k$ is a $4\times 4$ block defined as
\begin{equation}
 J_k:=\sum_{q=0}^{n-1} \left(
 \begin{array}{cc}
  \cos(2\pi kq/n) & -\sin(2\pi kq/n)\\
  \sin(2\pi kq/n) & \cos(2\pi kq/n)
 \end{array}
 \right) \otimes h_q.
\end{equation}

On the other hand, if $p=1$, one similarly proves that $H'':=({\cal R}_{2n} \otimes {\mathbbm{1}_2}) \tilde{H} ({\cal R}_{2n} \otimes {\mathbbm{1}_2})$:
\begin{equation}
 H''=\left(\bigoplus_{k=1}^{\lfloor n/2\rfloor} J_k'\right) \bigoplus \left[\sum_{q=0}^{n-1} (-1)^qh_q\right],
 \label{eq:app:bigoplusodd}
\end{equation}
where the last subspace only appears if $n$ is odd, and each $J_k'$ is a $4\times 4$ block defined as
\begin{equation}
 J_k':=\sum_{q=0}^{n-1} \left(
 \begin{array}{cc}
  \cos\left(\pi q\frac{2k-1}{n}\right) & -\sin\left(\pi q\frac{2k-1}{n}\right)\\
  \sin\left(\pi q\frac{2k-1}{n}\right) & \cos\left(\pi q\frac{2k-1}{n}\right)
 \end{array}
 \right) \otimes h_q.
\end{equation}

We can further simplify these expressions by noting that $H=-H^T$ is block-wise equivalent to $h_q=-h_{n-q}^T$. This implies that the $2\times 2$ blocks in Eqs. (\ref{eq:app:bigopluseven}) and (\ref{eq:app:bigoplusodd}) are already brought to the Williamson form:
\begin{equation}
 \sum_{q=0}^{n-1}h_q=\sum_{q=0}^{n-1}\left(
 \begin{array}{cc}
  0&(h_q)_{0;1}\\-(h_q)_{0;1}&0
 \end{array}
 \right),
\end{equation}
and
\begin{equation}
 \sum_{q=0}^{n-1}(-1)^qh_q=\sum_{q=0}^{n-1}(-1)^q\left(
 \begin{array}{cc}
  0&(h_q)_{0;1}\\-(h_q)_{0;1}&0
 \end{array}
 \right).
\end{equation}

We then define the quantity
\begin{equation}
\varepsilon_{0,\pm} = \sum_{q=0}^{n-1}(\pm 1)^q (h_q)_{0;1},
\end{equation}
which corresponds to the Williamson eigenvalue(s) for the $2\times 2$ blocks.

To bring $H$ to a Williamson form, it remains to bring the $4\times 4$ blocks $J_k$ and $J_k'$ to a Williamson form. To this end, let us start by defining $\upsilon_{k,q}:= q\pi (2k - (p+1)/2)/n$ and 
\begin{eqnarray}
 x_k&:=&\sum_{q=0}^{n-1} \cos(\upsilon_{k,q}) (h_q)_{0;1},\\
 a_k&:=&- \sum_{q=0}^{n-1} \sin(\upsilon_{k,q}) (h_q)_{0;0},\\
 b_k&:=&- \sum_{q=0}^{n-1} \sin(\upsilon_{k,q}) (h_q)_{0;1},\\
 c_k&:=&- \sum_{q=0}^{n-1} \sin(\upsilon_{k,q}) (h_q)_{1;1}.
\end{eqnarray}
Let us notice that the blocks $J_k$ (or $J_k'$) take the following form:
\begin{equation}
\left(
 \begin{array}{cccc}
  0&x&a&b\\
  -x&0&b&c\\
  -a&-b&0&x\\
  -b&-c&-x&0
 \end{array}
 \right)_k.
 \label{eq:JkBlock4x4}
\end{equation}
Now we are ready to find an orthogonal transformation that brings $J_k$ or $J_k'$ to a Williamson form, which we state in the following lemma:

\begin{lemma}
For every $J_k$ (or $J_k'$) of the form (\ref{eq:JkBlock4x4}), there exists an orthogonal transformation $O_k$ that brings it to a Williamson form
\begin{equation}
 O_k^T J_k O_k = \left(
 \begin{array}{cccc}
  0&\varepsilon_{k,+}&0&0\\
  -\varepsilon_{k,+}&0&0&0\\
  0&0&0&\varepsilon_{k,-}\\
  0&0&-\varepsilon_{k,-}&0
 \end{array}
\right).
\label{eq:JkinWilliamsonForm}
\end{equation}
Its two Williamson eigenvalues $\varepsilon_{k,\pm}$, are given by
\begin{equation}
\varepsilon_{k,\pm} = a_k + c_k \pm \sqrt{\Delta_k},
\end{equation}
where $\Delta_k := (a_k-c_k)^2 + 4 (b_k^2+x_k^2)$ and $k$ ranges from $1$ to $\lfloor n/2 + (p-1)/4\rfloor$.
Furthermore, this orthogonal transformation always satisfies $\det O_k = -1$.
\end{lemma}
\begin{proof}
The choice of $O_k$ is not unique in general. Here we are going to construct $O_k$ as the product of three matrices. Since $k$ is fixed, we are going to skip explicitly stating the subindex throughout the proof. We construct $O$ as $O:=L M R$, where $L$ and $R$ are diagonal matrices whose entries are defined by:
\begin{equation}
L^{-1}=\frac{1}{4}\mathrm{Diag}\{a-c-\sqrt{\Delta},2,a-c+\sqrt{\Delta},2\},
\end{equation}
and
\begin{equation}
R^{-1}=\sqrt{\frac{2}{b^2+x^2}}\mathrm{Diag}\{r_-,r_-,r_+,r_+\},
\end{equation}
where $r_{\pm}:=\sqrt{\Delta\pm (a-c)\sqrt{\Delta}}$.
The matrix $M$ in the middle is not diagonal, and in its most general form, can depend on two real parameters, which we denote $\phi$ and $\theta$. We have that the entries of $M$ are given by
\begin{equation}
\left(
\begin{array}{rrrr}
f_{b,-x}(\phi)g_+&f_{-x,b}(\phi)g_+&f_{-b,x}(\theta)&f_{x,b}(\theta)\\
\cos \phi&-\sin \phi & \cos \theta & -\sin \theta\\
-f_{x,b}(\phi) & f_{-b,x}(\phi) & f_{x,b}(\theta)g_-&f_{b,-x}(\theta)g_-\\
\sin \phi&\cos \phi & \sin \theta & \cos \theta
\end{array}
\right),
\end{equation}
where $$f_{y,z}(\alpha):=y\cos \alpha + z \sin \alpha,$$ and $$g_\pm := (c-a\pm \sqrt{\Delta})^2/4(b^2+x^2).$$
We can now show that $O$ is indeed an orthogonal transformation with determinant $-1$, since $\det{L} = -16/(b^2+x^2)$, $\det{R} = (b^2+x^2)/16\Delta$ and the determinant of $M$ is independent of both $\varphi$ and $\theta$, and it is $\det M = \Delta$. The multiplication of these three terms gives the result $\det O = -1$.
Hence, one can pick convenient values for $\phi$ and $\theta$ in order to show that $O^T J O$ already has a Williamson form (\ref{eq:JkinWilliamsonForm}) simply by matrix multiplication.
\end{proof}

Tracking all the transformations we have made, \textit{i.e.}, counting the parity flips imposed by the choice of all the orthogonal transformations we have made, the superselection rule Eq. (\ref{eq:superselectionrule}) that must be obeyed takes the form of Eq. (\ref{eq:megahypersuperselectionrule}).

If, in addition, we impose a finite interaction range $R$, then we can further simplify the expressions for $x_k$, $a_k$, $b_k$ and $c_k$ thanks to the property $h_r= -h_{n-r}^T$ and arrive at Eqs. (\ref{eq:michelangelo}, \ref{eq:donatello}, \ref{eq:raphael}, \ref{eq:leonardo}).

\section{\label{APP:TIIneq} Tight translationally invariant inequalities}
\label{app:tables}
In this appendix we provide a list of tight optimal TI Bell inequalities (\ref{eq:BIRis2}) for $R=2$ that are violated when performing the same set of measurements on each party, which are of the form
\begin{eqnarray}
 {\mathcal M}_0 &=& \cos \varphi \sigma_x + \sin \varphi \sigma_y,\nonumber\\
 {\mathcal M}_1 &=& \cos \theta \sigma_x + \sin \theta \sigma_y,\nonumber\\
 {\mathcal M}_2 &=& \sigma_z.\nonumber
\end{eqnarray}

Note that, since
\begin{equation}
 \mathrm{Tr}({\cal B} \rho) = \mathrm{Tr}(({\cal U}{\cal B}{\cal U}^\dagger)({\cal U}{\rho}{\cal U}^\dagger)),
 \label{eq:TUBU}
\end{equation}
where ${\cal U}$ is a unitary transformation of the form $U^{\otimes n}$, the maximal quantum violation of a TI Bell inequality with the above measurements only depends on $\theta - \varphi$, as there always exists a unitary $U$ that brings ${\mathcal M}_0$ to $\sigma_x$ by performing a rotation in the $x-y$ plane \cite{AnnPhys}. Hence, there is no loss of generality in assuming $\varphi = 0$.

For $n\leq 8$ it is still computationally feasible to find all the facets of the polytope of local correlations projected onto the space of the correlators appearing in Eq. (\ref{eq:BIRis2}).
To achieve this goal, we construct all the vertices of the local polytope of correlations for $n$ parties, $3$ measurements and $2$ outcomes, which are $2^{3n}$ in total, and we project them to the space of translationally invariant correlators that appear in (\ref{eq:BIRis2}), following the same procedure of \cite{TI50years}. With linear programming, we can remove all the projected vertices that are a convex combination of other projected vertices. Then, we can use an algorithm such as CDD \cite{cdd} to compute the convex hull of the extremal projected vertices and obtain a minimal description of it in terms of facets. We call these facets \emph{tight} Bell inequalities. We summarize these findings in Table \ref{tab:facets}:
\begin{table}[h]
\begin{center}
\begin{tabular}{ c|r|r} 
  $n$&\mbox{Number of facets}&\mbox{Number of vertices}\\
  \hline
  3&166&72\\
  4&5628&204\\
  5&46804&1148\\
  6&20268&1816\\
  7&175444&6064\\
  8&29290&4044
\end{tabular}
\end{center}
\caption{Number of tight Bell inequalities (facets) and number of extremal points (vertices) of the local polytope projected onto the space of TI, $n$-party, up to $R$-range correlators (cf. Eq. (\ref{eq:BIRis2})).}
\label{tab:facets}
\end{table}

Following the same procedure as in Example \ref{ex:8}, we find the inequalities that are violated, which we classify in the Table \ref{tab:numbers}.

Note that by renaming the outcomes of the measurements, the labels of the measurements or the labels of the parties, one can obtain other inequalities that are not listed in the table; however, these relabellings do not change the properties we are interested in, such as the classical bound or its quantum violation, so we consider them to be equivalent and we include only one representative for each equivalence class.

From the values of Table \ref{tab:numbers}, we see that $QV/\beta_C$ approaches zero as $n$ grows. Similar to what was found in \cite{TI50years}, numerics suggest that, for translationally invariant Bell inequalities, there is a trade-off between $n$ and $R$. In \cite{TI50years}, the maximal $n$ for $R=1$ was $5$, whereas here for $R=2$ we did not find any violation beyond $n=8$. Furthermore, for $n=6, 7$ there are no translationally invariant Bell inequalities of the form (\ref{eq:BIRis2}) that are violated by performing the same qubit measurements at each site. One may still perform different qubit measurements on each site and be able to maximally violate the inequality, as it was proven in \cite{TonerVerstraete}; however to achieve the same violation with the same set of measurements at each site, one may then need to increase the dimension of the state and use POVMs \cite{TI50years}.

Interestingly, we also note that the highest ratio $QV/\beta_C$ is obtained for $n=3$. This is not surprising, as for $n=3$ our inequalities contain full-body correlators (see Table \ref{tab:numbers}). However, for $n=4$, the first class achieves a much higher violation than the others. This is because the inequality consists of a sum of CHSH-like inequalities between parties $0$ and $2$ and between parties $1$ and $3$:
\begin{equation}
I=\sum_{i=0}^3 M_{(0,2,0)}^{(i,2)} - 2 M_{(0,2,1)}^{(i,2)} - M_{(1,2,1)}^{(i,2)}.
\end{equation}
The minimum over quantum values that can be achieved is $\beta_Q = -4\cdot 2\sqrt{2}$, so that $|\beta_Q|/\beta_C = \sqrt{2}$. This inequality can be generalized to any even number of parties at the price of increasing the interaction range $R$. By picking $R=n/2$ one can always pair party $k$ with party $k+R$ into a CHSH-like link while maintaining $|\beta_Q|/\beta_C = \sqrt{2}$ for any even $n$.
However, let us remark that for the scope of this work, we are interested in studying the nonlocality of ground states of local Hamiltonians (\textit{i.e.}, with a fixed $R$).

\section{\label{APP:NumericalOrtho} The orthogonal transformation for Majorana fermions}
We consider an antisymmetric matrix $H$ of size $2n$. If we want to decompose it as $H=O J O^T$, where $O$ is a $2n \times 2n$ orthogonal matrix and $J$ has the form (\ref{eq:Williamson}), we can in the majority of situations, use the Spectral Theorem to find $O$: The matrix $H^2$ is symmetric, so it diagonalizes as $H^2 = O D O^T$, with $D$ a diagonal matrix with entries $-\varepsilon_k^2$, appearing with multiplicity $2$ for each $k$, and the columns of $O$ forming an orthonormal basis. If all the $\varepsilon_k$ are different, then one can safely conclude that $H= O  J O^T$, because if $H$ is antisymmetric, then $J$ has to be of the form (\ref{eq:WilliamsonTI}). Hence, one can say that $O$ is unique (up to permutations that determine the order and the signs of $\varepsilon_k$).

However, if $\varepsilon_k$ has a multiplicity greater than $1$, this need no longer be the case, as the $O$ found via the Spectral Theorem is no longer unique (one can perform an arbitrary orthogonal transformation in each eigenspace). Note that this pathologic case is of interest to our problem, as for tight Bell inequalities with the optimal set of measurements, it is common to find $\varepsilon_k$'s with the same value (for instance, in Example \ref{ex:8}). There are two ways to circumvent this problem: one is to add some noise to $H$ such that all the $\varepsilon_k$'s can be considered different; however we lose precision in the solution and add numerical instability. The other way is described below.

Let $\varepsilon$ be a nonzero Williamson eigenvalue of $H$ with multiplicity $m$. Then, $O^T H O$ has a $2m \times 2m$ block, denoted $J_\varepsilon$, that satisfies
\begin{equation}
J_\varepsilon^2 = -\varepsilon^2 \mathbbm{1}_{2m}.
\label{eq:JBlockMultiplicity}
\end{equation}
Thus, any orthogonal transformation acting on $J_\varepsilon$ leaves $J_\varepsilon^2$ invariant, but not necessarily $J_\varepsilon$ in the form (\ref{eq:Williamson}). Let $\ket{e_1}$ be a unit vector in the $J_\varepsilon$ eigenspace (picked from the corresponding columns of $O$) and let $\ket{e_2}:=J_\varepsilon \ket{e_1}/\varepsilon$ be another unit vector. Then, we observe that $\ket{e_1} = -J_\varepsilon\ket{e_2}/\varepsilon$ because of Eq. (\ref{eq:JBlockMultiplicity}). Now, by picking another unit vector $\ket{e_3}$ from ${\mathbbm{1}}_{2m}-\ket{e_1}\bra{e_1} - \ket{e_2}\bra{e_2}$, we find $\ket{e_4}$ with the same procedure. We repeat this process $m$ times. By multiplying $O$ by the orthogonal transformation given by the vectors $\ket{e_1}, \ket{e_2}, \ldots \ket{e_{2m}}$ we obtain the right transformation bringing $H$ to the form (\ref{eq:Williamson}).

\onecolumngrid
\begin{table*}
\squeezetable
\begin{center}
\begin{tabular}{ c|r|rrrr|rrrr|r|c|c} 
$n$& $\gamma$ & $\gamma_{00}$& $\gamma_{01}$& $\gamma_{10}$& $\gamma_{11}$& $\gamma_{020}$& $\gamma_{021}$& $\gamma_{120}$& $\gamma_{121}$ &$\beta_C$ & QV & $\varphi_1-\varphi_0$\\
  \hline
  3&-2&1&-1&1&1&-1&-1&1&-1&6&-2.9282032303&$\pi/2$\\
  3&0&1&-1&3&-3&1&1&1&1&12&-2.5830052443&$\pi/2$\\
  3&2&-1&1&1&1&-1&1&1&1&6&-2.5830052443&$\pi/2$\\
  \hline
  4 & 0 & 0 & 0 & 0 & 0 & 1 & -2 & 0 & -1 & 8 & -3.313708499 & $\pi/2$\\
  4 & -2 & 4 & 2 & 2 & 0 & 5 & -3 & 1 & -1 & 32 & -0.5471047512 & $0.3254696365\pi$\\
  4 & 0 & 2 & 2& 2& 2& 5&-5& 1&-1&32&-0.5115214246 & $\pi/2$\\
  4&-1&2&5&5&6&-2&-10&2&11&72&-0.4999666746&$0.3188572387\pi$\\
  4&-1&-2&-3&-3&-4&-7&7&-1&2&48&-0.4670552431&$0.278787455\pi$\\
  4&2&0&-2&-2&-4&3&1&-3&-3& 32&-0.218521874 & $0.2029607403 \pi$\\
  \hline
  5&0&1&1&-1&-1&1&-1&1&-1&12&-0.3107341487& $\pi/2$\\
  \hline
  8&0&2&-2&2&-2&-1&-1&1&1&32&-0.2187&$\pi/2$
\end{tabular}
\end{center}
\caption{Classes of inequalities of the form of Eq. (\ref{eq:BIRis2}) that are violated for different $n$ with the same set of measurements at each site. We present one representative per class; the rest can be found by applying a suitable symmetry such as a renaming of the measurements, outcomes and/or parties \cite{TI50years}. The $\gamma$'s constitute the coefficients of the Bell inequality, $\beta_C$ is the classical bound, $QV$ is the quantum violation achieved with the measurement settings defined through $\varphi_0$ and $\varphi_1$. Note that due to Eq. (\ref{eq:TUBU}), this depends only on $\varphi_1-\varphi_0$ (Recall that the measurement settings we use parametrize the quantum observables as ${\cal M}_k^{(i)} = \cos \varphi_k^{(i)}\sigma_x^{(i)}+\sin \varphi_k^{(i)}\sigma_y^{(i)}$, and we omit the index $i$ because we are in the TI case).}
\label{tab:numbers}
\end{table*}
\twocolumngrid
\end{document}